\documentclass[12pt]{amsart}

\usepackage{geometry}

\newgeometry{
    top=0.5in,
    bottom=1.5in,
    outer=0.0in,
    inner=2in,
left=2.61cm,
right=2.61cm,
asymmetric
}
\usepackage[english]{babel}
\usepackage{amssymb}
\usepackage{amsfonts}
\usepackage{amsmath}
\usepackage{amsthm}
\usepackage{bm}
\usepackage{latexsym}
\usepackage{epsfig}
\usepackage{hyperref}
\usepackage{mathtools}
\numberwithin{equation}{section}
\newtheorem{theorem}{Theorem}[section]
\newtheorem{cor}[theorem]{Corollary}
\newtheorem{lemma}[theorem]{Lemma}
\newtheorem{proposition}[theorem]{Proposition}
\newtheorem{definition}[theorem]{Definition}

\newtheorem{remark}[theorem]{Remark}

\newtheorem{example}[theorem]{Example}
\newtheorem{question}[theorem]{Question}

\numberwithin{equation}{section}
\def\beq{\begin{equation}}
\def\eeq{\end{equation}}
\def\ben{\begin{enumerate}}
\def\een{\end{enumerate}}
\def\cB{{\mathcal B}}
\def\cC{{\mathcal C}}

\def\fm{{\mathfrak m}}

\def\fa{{\mathfrak a}}
\def\fb{{\mathfrak b}}
\def\cF{{\mathcal F}}

\def\cH{{\mathcal H}}
\def\cJ{{\mathcal J}}
\def\cL{{\mathcal L}}

\def\cQ{{\mathcal Q}}
\def\cN{{\mathcal N}}
\def\cM{{\mathcal M}}
\def\cP{{\mathcal P}}

\def\cS{{\mathcal S}}
\def\cT{{\mathcal T}}

\def\cV{{\mathcal V}}
\def\cW{{\mathcal W}}

\def\cW{\mathcal W}

\newcommand{\RR}{{\mathbb R}}
\newcommand{\PP}{{\mathbb P}}
\newcommand{\CC}{{\mathbb C}}

\newcommand{\ZZ}{{\mathbb Z}}

\newcommand{\NN}{{\mathbb N}}
\newcommand{\EE}{{\mathbb E}}
\newcommand{\lam}{{\lambda}}

\newcommand{\ntn}{n\times n}

\newcommand{\wt}{\widetilde}
\newcommand{\wh}{\widehat}
\newcommand{\ol}{\overline}
\newcommand{\ul}{\underline}

\newcommand{\fd}{\mathfrak{d}}

\usepackage{color}
\usepackage[normalem]{ulem}
\begin{document}

%+Title
\title{White Noise Space 
Analysis and Multiplicative Change of Measures}
\author[D. Alpay]{Daniel Alpay}
\address{(DA) Schmid College of Science and Technology\\
Chapman University,
One University Drive
Orange, California 92866,
USA.}
\email{alpay@chapman.edu}
\author[P.~Jorgensen]{Palle Jorgensen}
\address{(PJ) Department of Mathematics\\
14 MLH, The University of Iowa, 
Iowa City, Iowa 52242-
1419, USA.}
\email{jorgen@math.uiowa.edu}
\author[M. Porat]{Motke Porat}
\address{(MP) Department of Mathematics\\
Ben-Gurion University of the Negev\\
P.O. Box 653,
Beer-Sheva 84105\\ Israel}
\email{motpor@gmail.com}

\date{\today}
%-Title

%+Abstract
\begin{abstract}
In this paper we display a family of Gaussian processes, with explicit formulas and transforms. This is presented with the use of duality tools in such a way that the corresponding path-space measures are mutually singular. We make use of a corresponding family of representations of the canonical commutation relations (CCR) in an infinite number of degrees of freedom.

A key feature of our construction is explicit formulas for associated transforms; these are infinite-dimensional analogues of Fourier transforms. Our framework is that of Gaussian Hilbert spaces, reproducing kernel Hilbert spaces, and Fock spaces. The latter forms the setting for our CCR representations. We further show, with the use of representation theory, and infinite-dimensional analysis, that our pairwise inequivalent probability spaces (for the Gaussian processes) correspond in an explicit manner to pairwise disjoint CCR representations.
\end{abstract}
\maketitle
%-Abstract
%+Contents
\tableofcontents
%-Contents
\section*{Introduction}
Our framework is an infinite-dimensional 
harmonic analysis
and our main aims are fourfold. 
Starting with certain Gelfand triples, 
built over infinite-dimensional Hilbert space, 
we identify and construct, 
associated indexed families of infinite-dimensional probability spaces and corresponding stochastic processes. 
Second, we give conditions for when these stochastic processes are Gaussian. Third, we show that different values 
of the index-variable yield mutually 
singular probability measures. 
Fourthly, we identify associated representations 
of CCRs, and we show that pairwise different index values yield mutually inequivalent representations.  

Our current focus lies at the crossroads of white noise analysis, 
It\^o calculus, and algebraic quantum physics. While there is a very large prior literature in the area, we wish here to especially call attention to the following papers by Albeverio et al.
\cite{ADG,AM16,AM18,ASm,ARY,ASt}.

Some of the basic concepts that we study here can also be found,
in special cases, and in one form or the other, with variants of our main themes, at 
\cite{S56a,S56b,H80,J68}. Our present results go beyond those of earlier papers. This includes (i) our making an explicit and direct link between this wide family of stochastic processes, their harmonic analysis, and the corresponding representations. Also, (ii) our infinite-dimensional harmonic analysis, and infinite-dimensional transform theory, are new. Also new are (iii) our results which identify which of the general stochastic processes are Gaussian.

Our study of Gelfand triples in the framework of Schwartz tempered distributions is motivated in part by Wightman's  framework for quantum fields, i.e., that of (unbounded) operator valued tempered distributions. 
For details, see 
\cite{StWi00}.
\\
\\
We now turn to the study of stochastic processes indexed by an Hilbert space 
$\cH$
from a Gelfand triple
\begin{align*}
%\label{eq:11Dec20a}
\cS(\RR)\hookrightarrow
\cH
\hookrightarrow\cS^\prime(\RR),
\end{align*} 
side to side with the question of when such processes are Gaussian.
\iffalse
The framework for the Gelfand triples
which appear in this paper is obtained 
as follows: 
\fi
Here is a description on how to obtain the Hilbert space 
$\cH$. Consider the real Schwartz space
$\cS:=\cS(\RR)$ and a 
{\em conditional negative definite} (CND) function 
$\cN:\cS\rightarrow\RR$ 
that is continuous w.r.t 
the Fr\'echet topology and satisfies
$\cN(0)=0$.
\iffalse\textcolor[rgb]{0,1,0.501961}{such a function can be obtained from a stationary CND kernel on the Schwartz space that 
is continuous w.r.t the Fr\'echet topology.}
\fi
A general theory of Schoenberg then tells us  that 
$\cS$
can be isometrically embedded into an Hilbert space 
$\cH$, via a  mapping 
$s\mapsto\varphi_s$
such that 
$$\cN(s)=\|\varphi_s\|^2_{\cH}.$$
This idea goes back to works by
Schoenberg 
\cite{Sch37,Sch38a,Sch38b} and
von Neumann
\cite{SvN}, 
where the question 
"when can a metric space be realized 
in a Hilbert space with a
norm" was considered. 
The Hilbert space $\cH$ 
is taken to be the {\em reproducing kernel Hilbert space} (RKHS) 
with {\em reproducing kernel} (RK) 
of the form
\begin{align*}
\varphi_{\cN}(s_1,s_2)=
\big(\cN(s_1)+\cN(s_2)-\cN(s_1-s_2)\big)/2;
\end{align*}
for more details and references on how to obtain the Hilbert space 
$\cH$ from the CND function 
$\cN$, see Section
\ref{sec:CND}.
It further follows naturally from the Bochner--Minlos theorem, 
see e.g., 
\cite{T77}, applied for the positive definite functions 
$\cQ_{\lam}(s)=\exp\{-\lam^2\cN(s)/2\}$ with $\lam>0$, 
the existence of a one-parameter family of Gaussian 
measures 
$\{\PP_{\lam}\}_{\lam>0}$ on 
$\cS^{\prime}$, which satisfies
$$\EE_{\PP_{\lam}}\big[ 
\exp\{iX_s\}
\big]=
\exp\{-\lam^2\cN(s)/2
\},$$
where 
$X_s:\cS^\prime\rightarrow \RR$
is given by the duality of the
spaces 
$\cS$
and
$\cS^\prime$, as 
$X_s(\omega)=\langle\omega,s\rangle$, and 
$\EE_{\PP_{\lam}}$ stands for the expectation w.r.t the measure
$\PP_{\lam}$.  

The stochastic process given by 
$\{X_s\}_{s\in\cS}$ 
plays a main role in our analysis, as one of the purposes in the paper is to explore its behavior w.r.t different measures (from a family of measures described below, cf. equation
(\ref{eq:14Oct19h})); in particular we are interested in the cases where it is Gaussian. We show that the Gaussian property of that process is equivalent to a scaling assumption on the CND function $\cN$, simply saying that 
$$\cN(\alpha s)=\alpha^2\cN(s)$$
for every $s\in\cS$ and $\alpha\in\RR$ (cf. Proposition
\ref{prop:16Oct19a}).
As our settings are quite general, 
one can consider many examples for the CND function $\cN$,  which provide us with different examples of Gaussian processes. In particular, 
one can choose the CND function 
$\cN(s)=\|s\|^2_{L_2(\RR)}$, 
so that we get the standard 
Brownian Motion, or another 
function to get the 
fractional Brownian Motion (cf. 
Examples \ref{ex:13Dec20a}-\ref{ex:4Dec20a})

There are many dichotomy results regarding Gaussian measures, such as  Kakutani's dichotomy theorem on infinite product measures
\cite{Kak48}, which
states that any two Gelfand triple measures on 
$\cS^\prime$ are either equivalent or mutual singular. 
Also in
\cite[Chapter 5]{H80} 
it is proven that any two measures, from a given one-parameter family of measures, 
are mutually singular; this  corresponds to the analysis in this paper with the choice of  
$\cN(s)=\|s\|^2_{L_2(\RR)}$.
Another interesting result in that spirit is presented in
\cite{J68}, where explicit conditions for two measures being equivalent or mutually singular are given and 
the theory of RK Hilbert spaces is heavily used. 
By using the theory in 
\cite{J68}
we obtain the first main result in the paper
(cf. Theorem
\ref{thm:28Aug18d}), that is
establishing the crucial
fact of the measures 
$\{\PP_{\lam}\}_{\lam>0}$
being mutually singular.
To do so, we
explore the covariance function
$\Gamma_{\lam}(s_1,s_2)$ 
of the Gaussian process 
$\{X_s\}_{s\in\cS}$ 
w.r.t the measure 
$\PP_{\lam}$, 
that is 
$$\Gamma_{\lam}(s_1,s_2)=
\lam^2
\big( \cN(s_1+s_2)-\cN(s_1-s_2)\big)/4,$$
and compare between such 
covariance functions for 
distinct values of 
$\lam$.
Next, we make another important connection between the Hilbert space 
$\cH$ and the family of Gaussian measures
$\{\PP_{\lam}\}_{\lam>0}$. For that we let 
$\cH_{\lam}$
be the scaled Hilbert space, that is the space 
$\cH$ with the scaled norm
$\|\cdot\|_{\cH_{\lam}}=\lam\|\cdot\|_{\cH}$,
and then present an explicit transform 
(cf. equation
(\ref{eq:8May20c}))
$$W_{\lam}:\Gamma_{sym}(\cH_{\lam})\rightarrow L_2(\cS^\prime,\PP_{\lam})$$
that is an isometric isomorphism from the symmetric Fock space of 
$\cH_{\lam}$ 
onto 
$L_2(\cS^\prime,\PP_{\lam})$
(cf. Theorem
\ref{thm:22Apr19a}).
By using this transform we will be able to move from the notion of equivalence
of the Gaussian spaces (or more precisely the Gaussian measures) to the equivalence of the representations of the CCR algebra
on the symmetric Fock spaces; see the proof of Corollary 
\ref{cor:23Sep20a}. 
Another explicit formula which fits in naturally in the analysis presented in this paper is a
generalized (infinite dimensional) Fourier transform (cf. equation
(\ref{eq:19Oct19d})) , which takes
$L_2(\cS^\prime,\PP_{\lam})$ 
onto the RKHS that is given by 
the RK
$$\cQ_{\lam}(s_1,s_2)= 
\exp\{-\lam^2\cN(s_1-s_2)/2\}.$$
The last transform is  a generalization of results from
\cite{H80} 
to a wider family of conditional negative definite functions, other than
$\cN(s)=\|s\|^2_{L_2(\RR)}$.

Finally, we present the second main result of the paper (cf. Corollary 
\ref{cor:23Sep20a})
which is an explicit example of 
unitarily inequivalent
representations of the CCR algebra of an infinite dimensional Hilbert space
$\cH$; for details on the CCRs and representations, 
see e.g., 
\cite{S56a,S56b,Sh62}.
This example is built in a natural way from our analysis,
while due to the Stone--von Neumann theorem
 it is impossible to do so in the finite dimensional case. In the proof we use the  
mutual singularity of the measures from the family
$\{\PP_{\lam}\}_{\lam>0}$
and the intertwining operator 
$W_{\lam}$, 
to deduce that any two representations are disjoint, as the corresponding
measures are mutually singular. 
\\
\\
The outline of the paper is as follows: Section 
\ref{sec:background} is devoted to present some preliminaries results and definitions, which include positive definite kernels, conditional negative definite functions, the Bochner--Minlos theorem and (symmetric) Fock spaces.

In Section 
\ref{sec:CND}
we lay down our basic setting of the paper; we present in details some of the results by Schoenberg, connect those with the way we build our associated white noise space from the CND function
$\cN$ we started with, and establish the condition on $\cN$ that is equivalent for the studied stochastic process to be Gaussian (Proposition
\ref{prop:16Oct19a}).
Towards the end of the section we present some examples which fit into our setting.

In Section
\ref{sec:families}
we then build a family of $L_2$ path spaces which correspond to the family of Gaussian measures
$\{\PP_{\lam}\}_{\lam>0}$ obtained from the previous section after the simple change of multiplication by a scalar $\lam>0$. Then in Subsection
\ref{subsec:Jor}
we prove the mutual singularity of the measures (Theorem
\ref{thm:28Aug18d}) and give another example of a one-parameter family of mutually singular measures which involves the Fourier transform (Example
\ref{ex:12Dec20a}).

In Section
\ref{sec:IsomIsom}
we further study the intertwining operators between the symmetric Fock space and the $L_2$ space from the previous section (Theorem
\ref{thm:22Apr19a}), while also define an infinite dimensional generalized Fourier
transform and explore its properties 
(Lemma
\ref{lem:23Sep20a} and Theorem
\ref{thm:24Oct19a}) . Using these intertwining operators and the mutual singularity of the measures from previous sections, we establish the connection to representations of the CCR algebra of the (infinite dimensional Hilbert space) 
$\cH$ 
(Corollary 
\ref{cor:23Sep20a}).

\section{Background and Tools}
\label{sec:background}
We begin with presenting some preliminary background and definitions.
Let $X$ be a non-empty set. A function $\varphi:X\times X\rightarrow\RR$
is called {\em positive definite (PD)} kernel, if 
$\varphi$ is symmetric (i.e., $\varphi(x,y)=\varphi(y,x)$ for all $x,y\in
X$) and 
$$\sum_{j=1}^n
\sum_{k=1}^n
c_j\ol{c_k}\varphi(x_j,x_k)\ge0$$
for all 
$n\in\NN,\,
x_1,\ldots,x_n\in X$
and $c_1,\ldots,c_n\in\CC$.
A function 
$f:X\rightarrow\RR$ is called 
{\em positive definite (PD)}, if the
kernel 
$k_f(x,y)=f(x-y)$ 
is a positive definite kernel, i.e.,
if 
$f$ is odd and
\begin{align}
\label{eq:14Oct19c}
\sum_{j=1}^n\sum_{k=1}^n
c_j\ol{c_k}f(x_j-x_k)\ge0
\end{align}
for all 
$n\in\NN,\,
x_1,\ldots,x_n\in X$
and $c_1,\ldots,c_n\in\CC$.
A function
$\varphi:X\times X\rightarrow\RR$ is called 
{\em conditional negative
definite (CND)} kernel, if
$$\sum_{j,k=1}^n
c_j\ol{c_k}\varphi(x_j,x_k)\le0$$
for all
$n\in\NN,\,
x_1,\ldots,x_n\in X$ and 
$c_1,\ldots,c_n\in\CC$ 
with
$\sum_{j=1}^n c_j=0$.
Similarly to the positive definite definition, a function
$f: X\rightarrow \RR$ 
is called {\em conditional negative 
definite (CND)}, if
$\varphi_f(x,y):=f(x-y)$ is a CND kernel. Our conventions regarding CND functions follow
\cite{BCR}.

The class 
$\cS:=\cS(\RR)\subset L^2(\RR)$ 
is the Schwartz class of functions on 
$\RR$ 
which are rapidly decreasing, smooth and 
$\cC^\infty$.
The dual space of $\cS$ 
is the space 
$\cS^\prime:=\cS^\prime(\RR)$
of tempered distributions ,while any 
$\omega\in\cS^\prime$ 
defines a linear functional  
$\langle\omega,\cdot
\rangle:\cS\mapsto\RR$, 
 also called the action of 
$\omega$ 
on elements from 
$\cS$,
denoted by 
$\omega(s)=\langle\omega,s\rangle$
for 
$s\in\cS$. 
We recall the Bochner--Minlos 
theorem (associated to Gelfand triples with 
$\cS$
and
$\cS^\prime$) as it appears in 
\cite[Appendix A]{HOUZ}; 
for a more general framework of the theorem associated to nuclear Gelfand triples, see
\cite{T77}.
\begin{theorem}[Bochner--Minlos]
\label{thm:26Apr19a}
If 
$g:\cS\rightarrow\CC$ 
is continuous w.r.t the Fr\'echet topology, $g(0)=1$ and $g$ is positive definite,
in the sense of
(\ref{eq:14Oct19c}), 
then there exists a unique probability measure 
$P$ 
on 
$(\cS^\prime,\cB(\cS^\prime))$ 
such that\begin{align*}
\EE_{P}\big[\exp\{i\langle
\cdot,s\rangle\}\big]:=
\int_{\cS^\prime}
\exp\{i\langle\omega,
s\rangle\}dP(\omega)
=g(s),\,
\forall s\in\cS.
\end{align*}
We recall that
$L_2(\cS^\prime,\PP)$
is the
$L_2-$space for the 
white noise process.
\end{theorem}
For every 
$s\in\cS$, 
we define the random variable
$X_s$ on 
$\cS^\prime$ 
via the duality of 
$\cS$
and
$\cS^\prime$, that is by
$X_s(\cdot)=\langle
\cdot,s\rangle$, i.e.,
\begin{align*}
X_s(\omega)=\langle \omega,s\rangle
=\omega(s)
\end{align*} 
for every 
$\omega\in\cS^\prime$. 
One of the main studied objects 
in this paper is stochastic processes
$\{X_s\}_{s\in\cS}$ 
of that form and in particular the 
cases where those processes are Gaussian, 
while noticing that using the 
Bochner--Minlos theorem, different 
positive definite functions  
produce different probability measures on 
$\cS^\prime$ 
and hence different processes.
\begin{remark}
We use the framework of Gelfand triples 
which consist of an Hilbert space together with 
$\cS$ 
and 
$\cS^\prime$, 
see 
\cite{Sch64}. 
There are other works in infinite 
dimensional analysis, which use different 
approaches; one of those is to use 
$\cB$ and 
$\cB^\prime$ 
instead of
$\cS$ 
and 
$\cS^\prime$,
where $\cB$ 
is a Banach space and 
$\cB^\prime$ 
is its dual space,  see 
\cite{Gr67,Gr70}.
\end{remark}
Gaussian processes are one of the main objects 
in this paper, thereby we recall the 
explicit definition of a stochastic 
process being Gaussian. A stochastic process
$\{Y_t\}_{t\in I}$ 
is called a 
{\em Gaussian process (w.r.t a measure 
$\PP$)}, 
if for every 
$t_1,\ldots,t_n\in I$, 
the random vector
$Y_{t_1,\ldots,t_n}
=(Y_{t_1},\ldots,Y_{t_n})$ 
is a multivariate Gaussian
random variable; sometimes it is called a 
\textit{jointly} 
Gaussian process. 
This is equivalent to say that for every 
$\alpha_1,\ldots,\alpha_n\in\RR$
and 
$t_1,\ldots,t_n\in I$, 
the random variable 
$\alpha_1Y_{t_1}+\ldots+\alpha_n
Y_{t_n}$ 
has a univariate Gaussian distribution.
Another equivalent definition of 
$\{Y_t\}_{t\in I}$ 
being a Gaussian process (w.r.t 
$\PP$) 
--- using the characteristic function 
$\vartheta_{X_{t_1},\ldots,X_{t_n}}$ --- 
is that for every 
$t_1,\ldots,t_n\in I$, 
there exist
$\sigma_{j\ell},\mu_{\ell}\in\RR$ 
with 
$\sigma_{jj}>0$, 
such that
\begin{align}
\label{eq:22Oct19a}
\vartheta_{X_{t_1},\ldots,X_{t_n}}
(\alpha_1,\ldots,\alpha_n):=
\EE_{\PP}\Big[
\exp\Big\{i\sum_{\ell=1}^n
\alpha_{\ell}Y_{t_{\ell}}
\Big\}
\Big]
=\exp\Big\{i\sum_{\ell=1}^n
\mu_{\ell}\alpha_{\ell}-
2^{-1}\sum_{j,\ell=1}^n
\sigma_{j\ell}\alpha_{\ell}\alpha_j
\Big\},
\end{align}
here $\sigma_{j\ell}$ and $\mu_{\ell}$ can be shown to be the covariances
and means of the variables in the process and the left hand side of 
(\ref{eq:22Oct19a}) is the characteristic function of the random vector
$(Y_{t_1},\ldots,Y_{t_n})$.
\iffalse
\\\\\textbf{On measures.}
Let 
$(\Omega,\cF)$ 
be a measure space. Two measures
$\mu_1$ 
and 
$\mu_2$ 
on 
$(\Omega,\cF)$ 
are called 
\textbf{mutually singular} 
if there exist 
$A,B\in\cF$ 
such that 
$A\cap B=\emptyset,\,
A\cup B=\Omega$ 
and
$\mu_1(E\cap A)
=\mu_2(E\cap B)=0$
for all $E\in\cF$.
In that spirit, an indexed family of measures
$\{\mu_{\alpha}\}_{\alpha\in I}$ is called
\textbf{singular}, if it consists of pairwise mutually singular measures,
i.e., for every $\alpha,\beta\in I$, the measures 
$\mu_{\alpha}$ and
$\mu_{\beta}$ are mutually singular, as long as $\alpha\ne\beta$.
\fi

Recall the definitions of the
Hermite polynomials
\begin{align}
\label{eq:3Dec20a}
h_n(x)=(-1)^n
e^{\frac{1}{2}x^2}\frac{d^n}
{dx^n}\big(
e^{-\frac{1}{2}x^2}\big),
\quad n\ge0
\end{align}
and the Hermite functions
$\xi_n(x)=\pi^{-1/4}
((n-1)!)^{-1/2}
e^{-x^2/2}
h_{n-1}(\sqrt{2}x)$,
for 
$n\ge1$.
The set 
$\{h_n\}_{n\ge1}$ 
is an orthogonal basis for 
$L^2(\RR,\mu_1)$, 
where 
$d\mu_1(x)=(2\pi)^{-1/2}
e^{-x^2/2}dx$,
with 
$$\int_{\RR}h_n(x)h_m(x)d\mu_1(x)=
(2\pi)^{-1/2}
\int_{\RR} h_n(x)h_m(x)
e^{-x^2/2}dx
=n!\delta_{n,m}$$
henceforth the set 
$\big\{
(n!)^{-1/2}h_n
\big\}_{n\ge1}$ 
is an orthonormal basis of 
$L^2(\RR,\mu_1)$.
Finally, we adopt the approach and notations as in
\cite{F89}, to recall the definition of the symmetric Fock space of an Hilbert space. 
For an Hibert (separable) space 
$H$ with an orthonormal basis 
$\{e_j\}_{j\ge1}$, 
let 
$\Gamma(H)$ 
be the full Fock space over 
$H$, 
given by
$$\Gamma(H):=\bigoplus_{k=0}^\infty
H^{\otimes k}=
\CC\oplus H\oplus(H\otimes H)\oplus\ldots,$$
where 
$H^{\otimes k}$
is the 
$k'$th tensor power of 
$H$ 
for 
$k\ge1$ 
and 
$H^{\otimes 0}=\CC$;
the space 
$H^{\otimes k}$ 
has the orthonormal basis 
$\{e_{j_1}\otimes\ldots\otimes e_{j_k}
\}_{j_1,\ldots,j_k\ge1}$ 
and the union of all such bases form a 
basis for 
$\Gamma(H)$. 
We are interested in the (boson) 
\textit{symmetric Fock space over}
$H$, 
that is the subspace 
$\Gamma_{sym}(H)\subseteq \Gamma(H)$  
which consists of all symmetric tensors. 
More precisely, let 
$S$ 
be the symmetrizer  
which projects 
$\Gamma(H)$ 
onto
$\Gamma_{sym}(H)$, 
that is given on 
$H^{\otimes k}$
by
$$S(u_1\otimes\ldots\otimes u_k)
=(k!)^{-1}
\sum_{\sigma\in S_k}u_{\sigma(1)}
\otimes\ldots\otimes u_{\sigma(k)},$$
where $S_k$ is the permutation group of 
$\{1,\ldots,k\}$.
Then, $\Gamma_{sym}(H)$ 
has the orthonormal basis 
$\{E_{\alpha}\}_{\alpha\in\ell_0^{\NN}}$,
where
\begin{align}
\label{eq:5Dec20c}
E_{\alpha}:=
(k!)^{1/2}(\alpha!)^{-1/2}
S(e_1^{\alpha_1}
\otimes e_2^{\alpha_2}
\otimes\ldots),\, 
\alpha=(\alpha_1,\alpha_2,
\ldots)\in\ell_0^{\NN},
\end{align}
with 
$|\alpha|=\sum_{j=1}^\infty
\alpha_j=k$; here 
$\ell_0^{\NN}
=\{\alpha=(\alpha_1,\alpha_2,\ldots)
:\exists k\in\NN\,\forall
j>k, \alpha_j=0\}$
and $\alpha!=
\prod_{j=1}^\infty \alpha_j!$.
\section{Conditional Negative Definite Functions}
\label{sec:CND}
Let
$N:\cS\times\cS\rightarrow\RR$ 
be a real CND kernel, continuous w.r.t the 
Fr\'echet topology, 
such that $N(0,0)=0$.
In this paper we focus on the case where $N$
is \textit{stationary}, in the sense that 
$N(s_1,s_2)$ is a function of 
$s_1-s_2$, i.e., for every 
$s_1,s_2$ and $s_3$ in $\cS$, 
we have 
$N(s_1+s_3,s_2+s_3)=N(s_1,s_2)=N(s_2,s_1)$.
It is easily seen that 
all such CND kernels are characterized as 
$$N(s_1,s_2)=\cN(s_1-s_2),\,
\forall
s_1,s_2\in\cS$$
where
$\cN:\cS\rightarrow \RR$ 
is a CND function that is continuous 
w.r.t the Fr\'echet topology,
satisfying
$\cN(0)=0$.
In the later discussions we begin from the function $\cN$ rather
than $N$.
\iffalse
Notice that the symmetry of $N$ implies that 
\begin{align}
\label{eq:22Oct19c}
\cN(s)=\cN(-s),\,\forall s\in\cS.
\end{align}
\fi

The theory of CND kernels goes back to Schoenberg and von Neumann, where they answered the 
question of when a metric space can be realized in a Hilbert space with a norm; see
\cite{Sch37,Sch38a,Sch38b} and
\cite{SvN}.  
For a better understanding of what 
CND kernels are, 
we give a very wide family of conditional negative definite kernels,
in the spirit of
\cite[Proposition 3.2]{BCR}, 
which corresponds nicely to the full 
answer given by Schoenberg and von Neumann.
\begin{proposition}
\label{prop:23Oct19a}
Let
$X$ be a non-empty set.
For any vector space
$V$, a mapping $\Phi: X\rightarrow V$ 
and a symmetric, 
bilinear positive semi-definite form
$\langle\cdot,\cdot\rangle_V
:V\times V\rightarrow \CC$
(not necessarily inner product), 
the kernel 
$$N(x,y):=\| 
\Phi(x)-\Phi(y)\|_{V}^2$$
is a CND kernel on 
$X\times X$,
where 
$\|\cdot\|_V$ 
is the semi-norm induced from $\langle\cdot,\cdot\rangle_{V}$.
\end{proposition}
\begin{proof}
Due to the properties of the product $\langle\cdot,\cdot\rangle_V$,
for every
$n\in\NN,\,
x_1,\ldots,x_n\in X,$ and 
$c_1,\ldots,c_n\in\CC$ with $\sum_{j=1}^n
c_j=0$, we have
\begin{multline*}
\sum_{j,k=1}^n c_j\ol{c_k}N(x_j,x_k)=\sum_{j,k=1}^n
c_j\ol{c_k}\|\Phi(x_j)-\Phi(x_k)\|^2
=\sum_{j,k=1}^n
c_j\ol{c_k}\|\Phi(x_j)\|^2_V
\\-2\sum_{j,k=1}^n c_j\ol{c_k}\langle 
\Phi(x_j),\Phi(x_k)\rangle_V
+
\sum_{j,k=1}^n
c_j\ol{c_k}\|\Phi(x_k) \|^2_V
=-\Big\langle 
\sum_{j=1}^n c_j\Phi(x_j),\sum_{k=1}^n
c_k\Phi(x_k)\Big\rangle_V\le0.
\end{multline*}
\end{proof}

Following the theory by Schoenberg
\cite{Sch37,Sch38a,Sch38b}, while some of his results are well summarized in 
\cite[Proposition 3.2]{BCR}, as $N$ is a CND kernel with 
$N(s,s)=0$ for all
$s\in\cS$, there is an embedding of $\cS$ into  an Hilbert space
$\cH$ (which can be thought as a linear subspace of the space of all 
continuous
mappings from $\cS$ to $\RR$), 
i.e., there exists a mapping 
$s\mapsto \varphi_s$ from $\cS$
to
$\cH$ such that
\begin{align}
\label{eq:14Oct19d}
\cN(s_1-s_2)=N(s_1,s_2)
=\|\varphi_{s_1}-\varphi_{s_2}\|^2
=\|\varphi_{s_1-s_2}\|^2_{\cH}  
\end{align}
for any
$s_1,s_2\in\cS$.
We give some details on this building: it is an easy exercise to verify that
the kernel
\begin{align*}
\varphi_{\cN}(s_1,s_2)=2^{-1}
[N(s_1,0)+N(0,s_2)-N(s_1,s_2)]
=2^{-1}
\big[ \cN(s_1)+\cN(s_2)-\cN(s_1-s_2)\big]
\end{align*}
is a positive definite kernel
on $\cS\times\cS$.
Then, for every 
$s\in\cS$ define the mapping 
\begin{align*}
%\label{eq:28Apr20c}
\varphi_{s}:\cS\rightarrow \RR,
\text{ 
by }
\varphi_s(\wt{s})=\varphi_{\cN}(s,\wt{s}),
\,\forall\wt{s}\in\cS
\end{align*} 
and let $\cH_0$ be the linear subspace of $\RR^{\cS}$ (the space of all functions
from $\cS$ to $\RR$) that is generated by
$\{\varphi_s:s\in\cS\}$.
On $\cH_0$ 
we define an inner product by
\begin{align}
\label{eq:4Dec20c}
\Big\langle \sum_{j=1}^n \alpha_j \varphi_{s_j},
\sum_{\ell=1}^n
\beta_{\ell}\varphi_{\wt{s_{\ell}}} 
\Big\rangle:=\sum_{j=1}^n\sum_{\ell=1}^n
\alpha_j
\beta_{\ell}
\varphi_{\cN}(s_j,\wt{s_{\ell}}),
\end{align} 
which makes 
$\cH_0$ a pre-Hilbert space, hence there exists a (unique) continuation of
$\cH_0$ to an Hilbert space $\cH$,
in which $\cH_0$ is dense and we have the equation
$N(s_1,s_2)=\|\varphi_{s_1}-\varphi_{s_2}\|^2_{\cH}$. 
\iffalse
In particular, from the linearity of the mapping $s\mapsto \varphi_s$, we get that
\begin{align}
\label{eq:16Oct19b}
\cN(\alpha s)=\|\varphi_{\alpha s}\|^2_{\cH}=\|\alpha\varphi_s\|^2_{\cH}
=\alpha^2\|\varphi_s\|^2_{\cH}=\alpha^2\cN(s).
\end{align}
\fi
Notice that the Hilbert space $\cH$ 
is the reproducing kernel Hilbert space associated to the positive definite
kernel 
$\varphi_{\cN}(\cdot,\cdot)$ on $\cS\times\cS$. 
The Hilbert space
$\cH$ plays an important role in our analysis which comes up later, especially
when invoking its symmetric Fock space.
\\\\
If $N$ is as above, it follows from
\cite[Theorem 2.2]{BCR}
that 
$$Q(s_1,s_2)=
\exp\{-N(s_1,s_2)/2\}=\exp\{
-\cN(s_1-s_2)/2\}$$ 
is a PD kernel on 
$\cS\times\cS$.
\begin{remark}
\label{rem:4Dec20b}
A great contribution by Schoenberg, which is written and proved nicely in
\cite[Theorem 2.2]{BCR},
actually tells us that 
$N$ 
is a CND kernel if and only if 
$Q_t(s_1,s_2):=\exp\{
-t N(s_1,s_2)\}$
is a PD kernel for every 
$t>0$. 
\end{remark}
Let us then define 
$\cQ:\cS\rightarrow\RR$ by
\begin{align}
\label{eq:14Oct19a}
\cQ(s):=Q(s,0)=\exp\{-\cN(s)/2\},
\,\forall
s\in\cS.
\end{align}
It is easy to see that 
$\cQ(\cdot)$ 
is continuous w.r.t the Fr\'echet topology, 
with
$\cQ(0)=1$
and that 
$\cQ$ 
is a PD function, that is
\begin{multline*}
\sum_{j=1}^n
\sum_{k=1}^n
c_j\ol{c_k}\cQ(s_j-s_k)=\sum_{j=1}^n
\sum_{k=1}^n
c_j\ol{c_k}\exp\{-N(s_j-s_k,0)/2 \}
=\sum_{j=1}^n
\sum_{k=1}^n
c_j\ol{c_k}\exp\{ 
-N(s_j,s_k)/2\}
\\=\sum_{j=1}^n
\sum_{k=1}^n
c_j\ol{c_k}Q(s_j,s_k)\ge0
\end{multline*}
for every $n\in\NN,\,
c_1,\ldots,c_n\in\CC$, and $s_1,\ldots,s_n\in\cS$.
The function $\cQ$  
given in 
(\ref{eq:14Oct19a}) 
satisfies the conditions in Theorem
\ref{thm:26Apr19a}, therefore there exists a unique probability measure
$\PP$
on
$\cS^\prime$ such that
\begin{align}
\label{eq:14Oct19b}
\EE_{\PP}\big[\exp\{iX_s\}\big]
=\cQ(s)=\exp\{-\cN(s)/2\},
\,\forall s\in\cS
\end{align}
where for every 
$s\in\cS$, the random variable
$X_s$ on 
$\cS^\prime$ 
is defined via the duality of 
$\cS$
and
$\cS^\prime$, given by
$X_s(\cdot)=\langle
\cdot,s\rangle$, i.e.,
\begin{align}
\label{eq:22Jul20a}
X_s(\omega)=\langle \omega,s\rangle
=\omega(s)
\end{align} 
for every 
$\omega\in\cS^\prime$ and the notation $\EE_{\PP}$ stands for the expectation
w.r.t $\PP$, i.e., 
$\EE_{\PP}\big[f\big]:=
\int_{\cS^\prime}f(\omega)d\PP(\omega)$
for functions $f$ on the space
$\cS^\prime$.
From 
(\ref{eq:14Oct19d}) it then follows that
\begin{align}
\label{eq:14Oct19e}
\EE_{\PP}\big[\exp\{iX_s\}\big]
=\exp\{ -\|\varphi_s\|^2_{\cH}/2\},\,
\forall s\in\cS.
\end{align}
\iffalse
\begin{remark}
We immediately have 
$\varphi_0=0$, however we have no special knowledge on further properties
of $\varphi$, such as linearity.
We also know that 
$\varphi$ is $1-1$ if and only if $\ker\cN=\{0\}$.
\end{remark}
\begin{remark}
\label{rem:14Oct19f}
As $\cS$ is isometric to a dense subspace of 
$\cH$, that is
$\cH_0$, we use the identification 
$s\leftrightarrow
\varphi_s$ 
and then treat $\cS$
as a subspace of 
$\cH$. 
If we do that, then 
(\ref{eq:14Oct19e})
becomes
\begin{align}
\label{eq:14Oct19g}
\EE_{\PP}\big[ e^{i\langle\cdot,s\rangle}\big]
=\exp\Big\{-\|s\|^2_{\cH}/2
\Big\}.
\end{align} 
Our analysis includes a much bigger family of positive definite functions
on
$\cS$, whereas 
$(\ref{eq:14Oct19g})$
is only one of them which corresponds to a special choice of 
$\cN$ (cf. Example
\ref{ex:4Dec20a}). 
\end{remark}
\fi
A main issue of this paper (cf. Subsection
\ref{subsec:Jor}) lies on the fact that 
$\{X_s\}_{s\in\cS}$
is a 
{\em Gaussian process}, 
hence determined by its mean value and covariance
functions,  heavily depended on the measure that we assign to
$\cS^\prime$. However the stochastic process
$\{X_s\}_{s\in\cS}$
is not always Gaussian. In the following proposition we 
show that the Gaussian property of this process is equivalent to some scaling property of the CND function $\cN$: 
\begin{proposition}
\label{prop:16Oct19a}
Let 
$\cN:\cS\rightarrow\RR$
be a CND function, that is
continuous w.r.t the Fr\'echet topology,
with 
$\cN(0)=0$
and let
$\PP$ 
be the corresponding measure on $\cS^\prime$ that is defined in
$(\ref{eq:14Oct19b})$.
Then the stochastic process 
$\{X_s\}_{s\in\cS}$ is a (centered) Gaussian process (w.r.t 
$\PP$) 
if and only if $\cN$ satisfies
\begin{align}
\label{eq:16Oct19b}
\cN(\alpha s)=\alpha^2\cN(s),\,
\forall 
\alpha\in\RR,s\in\cS.
\end{align}
In that case, 
$\cN(s)$ is the variance of 
$X_s$ (w.r.t $\PP$).
\end{proposition}
\begin{proof}
First, observe that for every 
$\alpha_1,\ldots,\alpha_n\in\RR,\,
s_1,\ldots,s_n\in\cS$,
and $\omega\in\cS^\prime$, we have
$$\Big(\sum_{\ell=1}^n\alpha_\ell X_{s_\ell}
\Big)(\omega)=\sum_{\ell=1}^n
\alpha_\ell\omega(s_\ell)=
\omega
\Big(\sum_{\ell=1}^n
\alpha_\ell s_\ell\Big)
=X_{\sum_{\ell=1}^n
\alpha_\ell s_\ell}(\omega),$$
i.e., 
$\alpha_1X_{s_1}+\ldots+\alpha_n X_{s_n}=
X_{\alpha_1s_1+\ldots+\alpha_ns_n}$.
Thus $\{X_s\}_{s\in\cS}$ 
is (jointly) Gaussian if and only if 
$X_s$ is a Gaussian random variable
for every $s\in\cS$.
The characteristic function of $X_s$ 
is given by
$\vartheta_{X_s}(\alpha):=\EE_{\PP}\big[
\exp\{i\alpha X_s\} \big],$ 
hence --- as $X_{\alpha s}=\alpha X_s$ ---
$$\vartheta_{X_s}(\alpha)=\EE_{\PP}\big[
\exp\{iX_{\alpha s}\}\big]
=\exp\{-\cN(\alpha s)/2\}.$$
Therefore, 
$X_s$ 
is a Gaussian random variable
with covariance 
$\cN(s)$ if and only if 
$\vartheta_{X_s}(\alpha)=\exp\{
-\alpha^2\cN(s)/2\}$, i.e., if and only if
(\ref{eq:16Oct19b})
holds.

As for the Gaussian process being centered, we need to show that 
$\EE_{\PP}\big[ X_s\big]=0$
for every 
$s\in\cS$.
For every 
$\alpha\in\RR$
and
$s\in\cS$, we have
$\alpha X_s=X_{\alpha s}$, therefore by using 
(\ref{eq:16Oct19b})
we obtain that
\begin{align}
\label{eq:28Apr20a}
\int_{\cS^\prime}\exp
\{i\alpha X_s(\omega)\}
d\PP(\omega)=
\exp\{-\alpha^2\cN(s)/2\},\,
\forall\alpha\in\RR.
\end{align}
By invoking the Taylor expansions  (as functions of $\alpha\in\RR$) of both
sides of
(\ref{eq:28Apr20a}), 
together with the fact that on the right hand side the function is odd, as
$\cN(-s)=(-1)^2\cN(s)=\cN(s)$, 
we get that all the even coefficients must vanish. 
However, the coefficient of 
$\alpha^1$ 
in the expression on the left hand side of
(\ref{eq:28Apr20a})
is equal to 
$i\EE_{\PP}\big[ X_s\big]$, hence $\EE_{\PP}\big[ X_s\big]=0$.
\end{proof}
\iffalse
\begin{remark}
(?) As it is shown in Proposition 
\ref{prop:16Oct19a}, 
it is not automatically that the process 
$\{X_s\}_{s\in\cS}$ 
is Gaussian and it depends on 
$\cN$ via condition 
(\ref{eq:16Oct19b}). 
A simple example that is not Gaussian 
is the Poisson process. See
\cite{SS68,SS70}
where an interesting family of Wiener processes 
(which return independent random variables for disjoint sets) is studied and includes both Gaussian and Poisson processes.
\end{remark}
\fi
We finish this section by showing two well-studied examples, which arise naturally from particular choices 
of $\cN$.
In all of the cases mentioned below,  
we deal with a real CND function, 
that is continuous w.r.t the Fr\'echet 
topology and satisfies condition
(\ref{eq:16Oct19b}). 

To present those examples, one must recall an important family 
of (tempered) measures on 
$\RR$, 
which includes the Lebesgue 
measure, that is defined by
\begin{align}
\label{eq:4Oct19a}
\cM:=
\Big\{\mu:
\mu
\text{ is a positive measures on }\RR,
\text{ such that }
\int_{\mathbb R}
\frac{d\mu(u)}{u^2+1}<\infty\Big\}.
\end{align}
For every 
$\mu\in\cM$ we get a dual pair of path space measures $\PP_{\mu}$
and
$\PP_{\wh{\mu}}$
on
$\cS^\prime$, presented below in Examples 
\ref{ex:13Dec20a}
and
\ref{ex:2Oct20a}; 
one we may think of them as 
an infinite dimensional Fourier 
duality (for path-space measures).
In many examples, e.g. 
in financial math, 
for the corresponding Gaussian process 
$\{X_s\}_{s\in\cS}$ 
as part of a given dual pair,  
we will find that 
$\{X_s\}_{s\in\cS}$ 
will have fat tail 
w.r.t one of the path-space 
measures 
($\PP_{\wh{\mu}}$ 
in our case) in a dual pair; 
as compared to the other. 
In financial math application, 
realization of fat tails is important; see e.g.
\cite{JT19}.
\begin{example}
\label{ex:13Dec20a}
For any measure
$\mu\in\cM$, let 
\begin{align}
\label{eq:24Sep20d}
\cN_{\mu}(s):=
\int_{\RR}|s(u)|^2d\mu(u)=\|s\|^2_{\mu}.
\end{align}
This corresponds to Proposition 
\ref{prop:23Oct19a} 
by choosing
$X=\cS,\,
V=\cS,\,
\Phi:\cS\rightarrow\cS$  
to be the identity mapping and the product
$\langle s_1,s_2\rangle_{\mu}:
=\int_{\RR}
s_1(u)s_2(u)d\mu(u)$.
In that case the measure 
$\PP$, obtained from
(\ref{eq:14Oct19e}), 
is denoted by 
$\PP_{\mu}$ 
and the covariance function of the Gaussian process
$\{X_s\}_{s\in\cS}$ 
(w.r.t $\PP_{\mu}$)
is given by
\begin{equation*}
\EE_{\PP_{\mu}}\big[ 
X_{s_1}X_{s_2}
\big]=
\int_{\mathcal S^\prime}
X_{s_1}(\omega)
X_{s_2}(\omega) d\PP_{\mu}
(w)=\int_{\mathbb R}
s_1(u)
s_2(u)
d\mu(u),\,
\forall s_1,s_2\in\cS.
\end{equation*}
Then the Gaussian process can be re-indexed in 
$\RR_+$ instead of 
$\cS$, by the following rule:
for any
$t_1,t_2>0$, let
$s_1(u)=1_{[0,t_1]}(u)$
and
$s_2(u)=1_{[0,t_2]}(u)$, 
thus the covariance function is given by
\begin{align}
\label{eq:25Dec20a}
\int_{\mathcal S^\prime}
\langle w,1_{[0,t_1]}\rangle\langle w,1_{[0,t_2]}
\rangle
d\PP_{\mu}(w)=
\mu\big([0,\min\{t_1,t_2\}]\big).
\end{align}
See
\cite{AAL10,AAL12,AJL11} 
for more details.
\end{example}
To present the next example, 
we recall the definition of the Fourier transform 
\begin{align}
\label{eq:28Apr20b}
\wh{s}(t)=\int_{-\infty}^{\infty}
e^{-itu}s(u)du,
\quad\forall s\in L_2(\RR).
\end{align}
\begin{example}
\label{ex:2Oct20a}
For any measure
$\mu\in\cM$, let 
\begin{align}
\label{eq:24Sep20q}
\cN_{\wh{\mu}}(s):=
\int_{\RR}|\wh{s}(u)|^2d\mu(u)=\|\wh{s}\|^2_{\mu},
\end{align}
where $\wh{s}$
is the Fourier transform of $s$ (cf.
(\ref{eq:28Apr20b})).
This corresponds to Proposition 
\ref{prop:23Oct19a} 
by choosing
$X=\cS,\,
V=\cS,\,
\Phi:\cS\rightarrow\cS$  
to be the identity mapping and the product
$$\langle s_1,s_2\rangle_{\wh{\mu}}:
=\int_{\RR}
\wh{s_1}(u)\ol{\wh{s_2}(u)}d\mu(u).$$
In that case the measure 
$\PP$ is denoted by 
$\PP_{\wh{\mu}}$ 
and the covariance function of the Gaussian process
$\{X_s\}_{s\in\cS}$ (w.r.t $\PP_{\wh{\mu}}$)
is given by
\begin{equation*}
\EE_{\PP_{\wh{\mu}}}\big[ 
X_{s_1}X_{s_2}\big]=
\int_{\mathcal S^\prime}
X_{s_1}(\omega)
X_{s_2}(\omega) d\PP_{\wh{\mu}}(w)=\int_{\mathbb R}
\widehat{s_1}(u)
\overline{\widehat{s_2}(u)}
d\mu(u),\,
\forall s_1,s_2\in\cS.
\end{equation*}
Then the Gaussian process can be re-indexed in 
$\RR_+$ instead of 
$\cS$, by the following rule:
for any
$t_1,t_2>0$, let
$s_1(u)=1_{[0,t_1]}(u)$
and
$s_2(u)=1_{[0,t_2]}(u)$, 
thus the covariance function is given by
\begin{align}
\label{eq:25Dec20b}
\int_{\mathcal S^\prime}
\langle w,1_{[0,t_1]}\rangle\langle w,1_{[0,t_2]}
\rangle
d\PP_{\wh{\mu}}(w)=
\int_{\mathbb R}
\frac{e^{it_1u}-1}{u}
\frac{e^{-it_2u}-1}{u}d\mu(u);
\end{align}
see
\cite{AJ12,AJS14} for more details. 
\end{example}
We remind the reader that ---
in view of the two examples above --- 
for any 
$\mu\in\cM$,
we obtain two stochastic processes which are both given by the functions
$\{X_s\}_{s\in\cS}$, 
however they are taken to be w.r.t the different path-space measures 
$\PP_{\mu}$
and
$\PP_{\wh{\mu}}$.
Notice that the latter process is the generalized Fourier transform of the first process.
\begin{example}
\label{ex:17Nov20a}
The special case of the fractional Brownian motion ---
which has stationary independent increments but not independent increments ---
is considered, if one takes 
$\mu(u)=|u|^{1-2H}du$
for some
$H\in(0,1)$ 
and then (after re-indexing from 
$\cS$ 
to 
$\RR$ as in Example
\ref{ex:2Oct20a}) 
the covariance function
of the Gaussian process
$\{X_s\}_{s\in\cS}$ (w.r.t $\PP_{\wh{\mu}}$)
is given by 
$|t_1|^{2H}+|t_2|^{2H}-|t_1-t_2|^{2H}$.
For more details, see
\cite{AAL10,AK13}. 
\end{example}
\begin{example}
\label{ex:4Dec20a}
The two functions 
$\cN_{\mu}$ 
and
$\cN_{\wh{\mu}}$ 
coincide when $\mu$
is taken to be the Lebesgue measure (which is clearly in 
$\cM$). 
This will correspond to the 
classical Brownian motion.
For more details, see
\cite{H80}.
\end{example}
\iffalse
\begin{example}[exercise 1.20 in
\cite{BCR}]
Let 
$H$
be a \textbf{pre Hilbert} space, 
with the inner product 
$\langle\cdot,\cdot\rangle_{H}$ 
and the induced norm 
$\|\cdot\|_{H}$. Then the kernel given by the squared distance
$N(h_1,h_2):=\| h_1-h_2\|^2_{H}$ is a CND kernel; moreover it is clearly
a function of the difference $h_1-h_2$, hence $\cN(h):=\|h\|_{H}^2$ is a
CND function which satisfies all of our axioms.
\end{example}
\fi
\section{A One-Parameter Family of Gaussian Measures}
\label{sec:families}
In this section we study some behaviours of what happen when we do a simple
scaling by $\lam\in\RR_+$.
Let
$\cN:\cS\rightarrow\RR$
be a real CND function that 
is continuous w.r.t the Fr\'echet 
topology, which satisfies condition
(\ref{eq:16Oct19b})
and recall that we obtained the existence of an Hilbert space 
$\cH$
(cf. equation
(\ref{eq:14Oct19d})), with the special property that is $\cN(s)=\|\varphi_s\|^2_{\cH}$
for every $s\in\cS$.
\\\\\textit{\uline{A one-parameter family of Hilbert spaces}.}
For every 
$\lam\in\RR_+$,
we define the (scaled) Hilbert space
$\cH_{\lam}=(\cH,\langle\cdot,\cdot\rangle_{\cH_{\lam}})$, 
that is the space 
$\cH$ 
equipped with the (scaled) inner product
\begin{align}
\label{eq:5Dec20a}
\langle h_1,h_2\rangle_{\cH_\lam}
:=\lam^2\langle h_1,h_2
\rangle_{\cH},
\,\forall h_1,h_2\in\cH.
\end{align}
\textit{\uline{A one-parameter family of measures on
$\cS^\prime$}.}
For every 
$\lam\in\RR_+$, 
define
$\cQ_{\lam}:\cS\rightarrow\RR$ 
by
$$\cQ_{\lam}(s):=\cQ(\lam s)=\exp
\{-\cN(\lam s)/2 \},$$
which --- due to
(\ref{eq:16Oct19b}) ---
can be written as 
\begin{align}
\cQ_{\lam}(s)=\exp\{-\lam^2 \cN(s)/2\}.
\end{align}Due to 
\cite[Theorem 2.2]{BCR} and as it was explained in Remark
\ref{rem:4Dec20b},
$\cQ_{\lam}$ is PD thus  
we can apply the Bochner--Minlos theorem for 
$\cQ_{\lam}$. By doing so we get the existence (and uniqueness) of a probability
measure
$\PP_{\lam}$
on
$\cS^\prime$, such that
\begin{align}
\label{eq:14Oct19h}
\EE_{\PP_\lam}\big[
\exp\{iX_s\}\big]=\cQ_{\lam}(s)
=\exp\big\{
-\lam^2\|\varphi_s\|^2_{\cH}/2
\big\}=\exp\big\{
-\|\varphi_s\|^2_{\cH_\lam}/2
\big\}.
\end{align}
Notice that when 
$\lam=1$, 
we have that 
$\PP_1=\PP$
is the measure obtained in 
(\ref{eq:14Oct19b}).
\begin{remark}
\label{rem:22Oct19c}
Since
$\lam s\in\cS$ 
for every 
$s\in\cS$ and
$\lam\in\RR_+$, 
we have the following obvious
relation between any two measures from
$\{\PP_{\lam}\}_{\lam\in\RR_+}$, 
that is
\begin{align*}
\EE_{\PP_{\lam_1}}
\big[ \exp\{i\lam_2X_{s}\}\big]=
\cQ_{\lam_1}(\lam_2s)=
\exp\big\{-\lam_1^2\lam_2^2\cN(s)/2
\big\}=
\cQ_{\lam_2}(\lam_1s)=\EE_{\PP_{\lam_2}}\big[
\exp\{i\lam_1X_{s}\}\big],
\end{align*} 
for every 
$\lam_1,\lam_2\in\RR_+$ 
and 
$s\in\cS$.
\end{remark}
\begin{remark}
\label{rem:19Sep20a}
It is possible to have corresponding definitions for $\lam<0$, 
however the answers for all the  questions we study in this paper depend only on the value of
$\lam^2$.
\end{remark}
For every 
$\lam\in\RR_+$, 
we define the 
$\lam-$\textbf{white noise space} 
\begin{align}
\label{eq:23Sep20a}
L_2(\cS^\prime,\PP_{\lam}):=
\Big\{f:\cS^\prime
\rightarrow\CC\mid
\int_{\cS^\prime}|
f(\omega)|^2d\PP_{\lam}(\omega)<\infty 
\Big\},
\end{align}
which is an inner product space w.r.t the inner product
\begin{align}
\label{eq:23Sep20b}
\langle f,g
\rangle_{L_2(\cS^\prime,\PP_\lam)}:=
\int_{\cS^\prime}f(\omega)
\ol{g(\omega)}
d\PP_{\lam}(\omega),\,
\forall f,g\in L_2
(\cS^\prime,\PP_{\lam}).
\end{align}
\begin{lemma}
The space $L_2(\cS^\prime,\PP_{\lam})$ 
contains all the functions 
$\{X_s\}_{s\in\cS}$,
while their moments are given in formulas  
(\ref{eq:28Aug18b})
and
(\ref{eq:22Apr19g}).
\end{lemma}
As
$\{X_s\}_{s\in\cS}$ 
is a centered Gaussian process,
its odd moments vanish, while its even 
$2n-$moments are equal to 
$(2n-1)!!\sigma_s^{2n}$, 
where 
$\sigma_s=\cN(\lam s)$ 
is the variance of 
$X_s$; 
for a proof, see e.g.
\cite[Proposition 6.2]{AJL17}. 
Nevertheless, we provide the readers a short proof.
\begin{proof}
As
$X_{\alpha s}=\alpha X_s$ 
and
$\|\varphi_{\alpha s}
\|_{\cH_{\lam}}^2
=\alpha^2\|\varphi_s\|^2_{\cH_{\lam}}$
for every
$\alpha\in\RR$ 
and
$s\in\cS$, 
we apply equation 
(\ref{eq:14Oct19h}) 
to obtain 
\begin{align*}
\,\exp\big\{-\alpha^2\|
\varphi_{ s}\|_{\cH_{\lam}}^2/2\big\}
=\cQ_{\lam}(\alpha s)=\EE_{\PP_\lam}
\big[\exp\{iX_{\alpha s}\}\big]
=\int_{\cS^\prime}
\exp\{i\alpha X_s(\omega)\}d\PP_{\lam}(\omega).
\end{align*}
%\begin{remark}
%This is where we have to assume that
%$\cN(\alpha s)=\alpha\cN(s)$ for any $\alpha>0$ and $s\in\cS$. Otherwise,
%we are kind of stuck.
%\end{remark}
Thus, by invoking the Taylor 
expansions which correspond to 
the exponents in both sides,
we have 
\begin{align*}
\sum_{n=0}^\infty \frac{(-1)^n
\|
\varphi_s\|_{\cH_{\lam}}^{2n}}{2^nn!}
\alpha^{2n}=
\sum_{n=0}^\infty \frac{i^n}{n!}
\Big(\int_{\cS^\prime}X_s(\omega)^n
d\PP_{\lam}(\omega) 
\Big)\alpha^n,
\end{align*}
so for every $n\in\NN_0$ 
and 
$s\in\cS$,
the $(2n+1)-$moment of 
$X_s$ w.r.t $\PP_{\lam}$ is given by
\begin{align}
\label{eq:28Aug18b}
\int_{\cS^\prime}X_s(\omega)^{2n+1}
d\PP_{\lam}(\omega)=0
\end{align}
and the $2n-$moment of ${X_s}$ w.r.t $\PP_{\lam}$ is given by
\begin{equation}
\label{eq:22Apr19g}
\int_{\cS^\prime}X_s(\omega)^{2n}
d\PP_{\lam}(\omega)
=(2n-1)!!
\|\varphi_s\|_{\cH_{\lam}}^{2n},
\end{equation}
where
$(2n-1)!!:=1\cdot 3\cdots(2n-1)=\frac{(2n)!}{2^nn!}$.
The last equation can be rewritten as
\begin{equation*}
\big\| X_s^{n}\big\|_{L_2
(\cS^\prime,\PP_{\lam})}^2
=(2n-1)!!
\|
\varphi_s\|^{2n}_{\cH_{\lam}},\,
\forall 
n\in\NN_0,\,
s\in\cS.
\end{equation*} 
In particular,
for any 
$s\in\cS$, 
we obtained that
$X_s\in L_2
(\cS^\prime,\PP_{\lam})$ with
\begin{equation}
\label{eq:22Apr19f}
\|X_s\|_{
L_2(\cS^\prime,\PP_{\lam})}=
\|
\varphi_s\|_{\cH_{\lam}},
\end{equation} 
while 
$\EE_{\PP_{\lam}}\big[ X_s\big]=0$ 
and the variance of
$X_s$
(w.r.t the measure $\PP_{\lam}$) 
is equal to
$\EE_{\PP_{\lam}}\big[ X_s^2\big]=\|X_s\|^2_{L_2(\cS^\prime,\PP_{\lam})}
=\|\varphi_s\|^2_{\cH_\lam}=
\lam^2\cN(s).$
\end{proof}
\begin{remark}
As we complete $\cS$ to the Hilbert space 
$\cH_{\lam}$,
we use the It\^o isometry to determine 
that once we index the Gaussian process by elements 
$h\in\cH_{\lam}$, the mapping 
$h\mapsto X_h$ 
is an isometry from the Hilbert space
$\cH_{\lam}$ 
to 
$L_2(\cS^\prime,\PP_{\lam})$. In other words, the Gaussian process 
$\{X_s\}_{s\in\cS}$
is now extended to the Gaussian process  
$\{X_h\}_{h\in\cH_{\lam}}$, with 
$$\EE_{\PP_{\lam}}\big[|X_h|^2\big]
=\|h\|_{\cH_{\lam}}^2$$ 
which is the extension of
(\ref{eq:22Apr19f}). 
\end{remark}
Next, we introduce the 
function
$\Gamma_{\lam}:\cS\times\cS\rightarrow\RR$ 
that is defined by
\begin{align}
\label{eq:21Oct19d}
\Gamma_{\lam}(s_1,s_2):=
\EE_{\PP_{\lam}}[X_{s_1}X_{s_2}\big]
=\langle X_{s_1},X_{s_2}\rangle_{L_2(\cS^\prime,\PP_{\lam})}
\end{align} 
for every
$s_1,s_2\in\cS$;
that is the covariance function of the Gaussian process
$\{X_s\}_{s\in\cS}$ 
w.r.t the measure 
$\PP_{\lam}$, so it is obviously PD. 
Notice that equation 
(\ref{eq:22Apr19f})
connects the two Hilbert spaces 
$L_2(\cS^\prime,\PP_{\lam})$ 
and 
$\cH_{\lam}$, via the norm preserving mapping 
$X_s\leftrightarrow\varphi_s$.
\begin{proposition}
The covariance function of the 
Gaussian process
$\{X_s\}_{s\in\cS}$ 
w.r.t the measure 
$\PP_{\lam}$ 
admits the formula
\begin{align}
\label{eq:23Oct19f}
\Gamma_{\lam}(s_1,s_2)=4^{-1}
\lam^2
\big(\cN(s_1+s_2)-\cN(s_1-s_2)\big),
\,\forall s_1,s_2\in\cS.
\end{align}
The mapping 
$X_s\rightarrow\varphi_s$ 
(from a subspace of 
$L_2(\cS^\prime,\PP_{\lam})$
into $\cH_{\lam}$) is an isometry 
if and only if 
\begin{align}
\label{eq:22Oct19d}
\cN(s_1)+\cN(s_2)=2^{-1}
\big( \cN(s_1+s_2)+\cN(s_1-s_2)\big).
\end{align}
\end{proposition}
This is an important fact for us, which will come into play in
Subsection 
\ref{subsec:Jor}.

\begin{proof}
By using the parallelogram law in 
$L_2(\cS^\prime,\PP_{\lam})$
\iffalse
$\langle u,v\rangle_{H}
=\frac{1}{4}(\|u+v\|_{H}^2-
\|u-v\|_{H}^2)$
that is true for every Hilbert space $H$ and $u,v\in H$,
\fi
and the fact that
$X_{s_1+s_2}=X_{s_1}+X_{s_2}$ for every $s_1,s_2\in\cS$, 
we know that 
\begin{multline*}
\langle X_{s_1},X_{s_2}
\rangle_{L_2(\cS^\prime,\PP_{\lam})}
=4^{-1}\big(
\|X_{s_1+s_2}\|^2_{L_2(\cS^\prime,
\PP_{\lam})}-\|X_{s_1-s_2}\|^2_{L_2
(\cS^\prime,\PP_{\lam})}\big)
\\=4^{-1}
\big(\|\varphi_{s_1+s_2}\|^2_{\cH_{\lam}}-\|
\varphi_{s_1-s_2}\|^2_{\cH_{\lam}}\big)
=4^{-1}\lam^2(\cN(s_1+s_2)-\cN(s_1-s_2)).
\end{multline*}
Therefore, the covariance function of the Gaussian process
$\{X_s\}_{s\in\cS}$ w.r.t the measure $\PP_{\lam}$ admits the formula
in
\ref{eq:23Oct19f}.
Moreover, as the inner product in 
$\cH_{\lam}$ is defined as follows 
(cf. equation 
(\ref{eq:4Dec20c}))
$$\langle\varphi_{s_1},\varphi_{s_2}\rangle_{\cH_{\lam}}
=\lam^2\varphi(s_1,s_2)=
2^{-1}\lam^2
\big(\cN(s_1)+\cN(s_2)-\cN(s_1-s_2)\big),$$
we obtain that
$\langle X_{s_1},X_{s_2}\rangle_{L_2(\cS^\prime,\PP_{\lam})}=
\langle \varphi_{s_1},\varphi_{s_2}\rangle_{\cH_{\lam}}$
if and only if the function 
$\cN$ 
meets condition
(\ref{eq:22Oct19d}).
So the mapping 
$X_s\rightarrow\varphi_s$ is always norm preserving (as
$\|X_s\|^2_{L_2(\cS^\prime,\PP_{\lam})}
=\|\varphi_s\|_{\cH_\lam}^2=\lam^2\cN(s)$), 
however it is isometry
if and only if  $\cN$ 
satisfies
(\ref{eq:22Oct19d}). 
\end{proof}
Finally, we recall a general formula for the joint distribution of
$X_{\xi_1},\ldots,X_{\xi_n}$,
where $\xi_1,\ldots,\xi_n\in\cS$ 
are linearly independent, which holds as the process 
$\{X_s\}_{s\in\cS}$
is Gaussian.
\begin{lemma}
\label{lem:21Oct19a}
Let $n\in\NN$
and
$\xi_1,\ldots,\xi_n\in\cS$ 
be linearly independent.
Define the (truncated covariance) matrix 
$C_n^{(\lam)}\in\RR^{\ntn}$ 
by 
$(C_n^{(\lam)})_{ij}=
\Gamma_{\lam}(\xi_i,\xi_j)=
\EE_{\PP_{\lam}}\big[ X_{\xi_i}
X_{\xi_j}\big].$ 
Then the joint distribution of 
$X_{\xi_1},\ldots,X_{\xi_n}$ 
w.r.t $\PP_{\lam}$ is given by the function
\begin{align}
\label{eq:21Oct19b}
g_n^{(\lam)}(\ul{x}):=
(\det C_n^{(\lam)})^{-1/2}\exp\Big\{
\ul{x}^T\big(
C_n^{(\lam)}\big)^{-1}\ul{x}\Big\}
\end{align}
where
$\ul{x}=(x_1,\ldots,x_n)$, i.e., for every 
$f\in L_1(\RR^n,g_n^{(\lam)}(\ul{x})d\ul{x})$
we have
\begin{align*}
\EE_{\PP_{\lam}}\big[ f(X_{\xi_1},\ldots,X_{\xi_n})\big]
=\int_{\RR^n}
f(\ul{x})g_n^{(\lam)}
(\ul{x})d\ul{x}.
\end{align*}
\end{lemma}
It is easily seen from
(\ref{eq:23Oct19f}) that for any 
$\lam_1,\lam_2\in\RR_+$, we have the following relation between the covariance functions, that is 
$$\lam_1^{-2}\Gamma_{\lam_1}(s_1,s_2)=
\lam_2^{-2}\Gamma_{\lam_2}(s_1,s_2),$$
which implies the following relation between the covariance matrices
$\lam_1^{-2}C_n^{(\lam_1)}=
\lam_2^{-2} C_n^{(\lam_2)}.$
The last equality implies that
the way that the covariance matrix
$C_n^{(\lam)}$ depends on
$\lam$ is just by a diagonal matrix, and also that 
 the joint distributions of the Gaussian processes
 (w.r.t $\PP_{\lam_1}$ 
 and 
 $\PP_{\lam_2}$) admit the following relation
$$\lam_1^ng_n^{(\lam_1)}(\lam_1x_1,\ldots,\lam_1x_n)
=\lam_2^ng_n^{(\lam_2)}(\lam_2x_1,\ldots,\lam_2x_n).$$
\begin{remark}
It is  obvious from the formula
(\ref{eq:21Oct19b}), that if 
$C_n^{(\lam)}$ 
is a diagonal matrix, then  the function
$g_n^{(\lam)}$ is of the special form
$g_n^{(\lam)}(\ul{x})=
\alpha_{0,\lam}\exp\{\alpha_{1,\lam}x_1^2\}
\cdots\exp\{\alpha_{n,\lam} x_n^2\},$
where $\alpha_{0,\lam},
\alpha_{1,\lam},\ldots,\alpha_{n,\lam}\in\RR$. 
This will allow us to provide an
easy way to build an orthonormal basis for the space
$L_2(\cS^\prime,\PP_{\lam})$. 
However, in order for $C_n^{(\lam)}$
to be diagonal, we need to choose
$\xi_1,\ldots,\xi_n\in\cS$ such that
\begin{align}
\label{eq:23Oct19e}
\Gamma_{\lam}(\xi_i,\xi_j)=0,
\text{ i.e., }\cN(\xi_i+\xi_j)
=\cN(\xi_i-\xi_j),\, 
\forall i\ne j.
\end{align}
In general, we know that 
$\Gamma_{\lam}(s_1,s_2)$ 
is a symmetric, positive semi-definite bilinear mapping on 
$\cS\times\cS$, 
therefore we might use the Gram--Schmidt algorithm to get a set
$\{\xi_i\}_{i=1}^n$ in 
$\cS$
such that
$\Gamma_{\lam}(\xi_i,\xi_j)=\delta_{ij}$. 
\end{remark}
\iffalse
For example, in the special 
cases where the CND function is 
$\cN(s)=\|s\|^2_{L_2(\RR)}$, 
as the condition in
(\ref{eq:23Oct19e})
simply becomes $\langle\xi_i,\xi_j\rangle=0$ for $i\ne j$ and can be justified
by the Gram--Schmidt algorithm. 
\fi
\subsection{Mutual singularity of the measures}
\label{subsec:Jor}
In
\cite[Chapter 3]{H80}
the author studies the behaviour of the space
$L_2(\cS^\prime,\mu_{\lam})$,
where the measure 
$\mu_{\lam}$ 
is the one obtained from the Bochner--Minlos 
theorem applied to the PD function
$$\exp
\big\{-\lam^2\|s\|^2/2
\big\}=\EE_{\mu_\lam}
\big[ \exp\{iX_s\}\big];$$ 
here 
$\|s\|$ 
stands for the
$L_2(\RR)$ norm, that is 
$\|s\|^2=\int_{\RR}s(u)^2du$.
In particular, it is shown
(see
\cite[Proposition 3.1]{H80}) that whenever
$\lam_1,\lam_2\in\RR_+$ 
and
$\lam_1\ne\lam_2$, the measures $\mu_{\lam_1}$
and
$\mu_{\lam_2}$
are mutually singular. 
This fits exactly into our 
settings, simply by choosing 
the CND function
$\cN(s)=\|s\|^2_{L_2(\RR)}$,
thus the measure 
$\mu_{\lam}$
which appears in
\cite[Chapter 3]{H80}
coincides with the measure 
$\PP_{\lam}$ 
introduced earlier (in
(\ref{eq:14Oct19h})).

The purpose of this subsection is to generalize 
\cite[Proposition 3.1]{H80}
to the case where the measures 
$\{\PP_{\lam}\}_{\lam\in\RR_+}$ 
correspond to a function 
$\cN:\cS\rightarrow\RR$
which is a CND function,
continuous w.r.t the Fr\'echet topology,
and satisfies condition 
(\ref{eq:16Oct19b}).
To do so, we will use the results from
\cite{J68}, 
which highly depend on the theory 
of RKHSs, whereas 
the leading idea is to compare between the covariance functions
of a Gaussian process w.r.t two different measures and establish a condition
which determines whether the measures are equivalent or singular. 

For the convenience of the reader, 
we first recall the settings from 
\cite{J68}. 
Let
$(\Omega,\cF)$ 
be a measurable space, where 
$\cF$ 
is the 
$\sigma-$algebra generated 
by a class of random variables 
$\{X(t):t\in T\}$
and 
$T$ 
is assumed to be an interval, or more generally a separable metric space.
Assume 
$\wt{P}$ 
and
$P$
are probability measures on
$(\Omega,\cF)$, 
such that 
$\{X(t):t\in T\}$ 
are Gaussian processes with 
mean value functions 
$\wt{m}(t)$ and 
$0$, 
and covariance functions 
$\wt{\Gamma}(s,t)$ 
and 
$\Gamma(s,t)$, 
respectively. Then, as they are PD kernels, the covariance
functions 
$\wt{\Gamma}(s,t)$
and
$\Gamma(s,t)$ 
generate RKHSs
$\cH(\wt{\Gamma})$
and
$\cH(\Gamma)$, 
with the RK
$\wt{\Gamma}(s,t)$ 
and 
$\Gamma(s,t)$, 
respectively. 
If
$\{g_k\}_{k=1}^\infty$
is a complete orthonormal system in
$\cH(\Gamma)$, then the RK has the following form
$$\Gamma(s,t)=\sum_{k=1}^\infty
g_k(s)g_k(t);$$
for a proof see e.g.
\cite[Proposition 3.8]{J68}. 
For a more elegant and general presentation of
the RKHSs $\cH(\Gamma)$ see
\cite{Sch64}, where Hilbert spaces of tempered distributions are discussed as well.
We will use the following result:

\begin{theorem}[Theorem 4.4  in 
\cite{J68}]
\label{thm:28Aug18c}
The measures 
$\wt{P}$ 
and
$P$
are mutually equivalent 
if and only if the following hold:
\begin{enumerate}
\item[1)]
$\wt{m}(\cdot)\in \cH(\Gamma)$,
\item[2)]
$\wt{\Gamma}$ 
has a representation
\begin{align}
\wt{\Gamma}(s,t)
=\sum_{k=1}^\infty 
\beta_k g_k(s)g_k(t),
\end{align}
where 
$\{g_k\}$ 
is a complete orthonormal system in 
$\cH(\Gamma)$,
with
\begin{align}
\sum_{k=1}^\infty
(1-\beta_k^2)<\infty
\text{ and }\beta_k>0
\text{ for all }k.
\end{align}
\end{enumerate}
\end{theorem} 
Our interpretation of 
the results from
\cite{J68} is well summarized in the proof of the  following theorem, 
where we consider the 
Gaussian process
$\{X_s\}_{s\in\cS}$
and the distinct path-space measures 
are any pair of measures from 
the family 
$\{\PP_{\lam}\}_{\lam\in\RR_+}$,
introduced in
(\ref{eq:14Oct19h}). 
\begin{theorem}
\label{thm:28Aug18d}
Let 
$\cN:\cS\rightarrow\RR$
be a CND function, that is
continuous w.r.t the Fr\'echet topology
and satisfies condition
(\ref{eq:16Oct19b}).
Then the measures 
$\PP_{\lam_1}$ 
and 
$\PP_{\lam_2}$ 
are mutually singular,
for every 
$\lam_1,\lam_2\in\RR_+$ with 
$\lam_1\ne\lam_2$. 
\end{theorem}
\begin{proof}
Fix $\lam_1,\lam_2\in\RR_+$.
Let our measurable space
$(\Omega,\cF)$ be given by
$\Omega=\cS^\prime$
and
$\cF
=\cB(\cS^\prime)$,
while the probability measures be 
$\wt{P}=\PP_{\lam_1}$ 
and 
$P=\PP_{\lam_2}$, 
as obtained in
(\ref{eq:14Oct19h}),
hence the dependence in 
$\cN$.
Let
$T=\cS$
which is a separable metric space
and consider
the random variables
$X(s)
=X_s$
for every
$s\in\cS$.
The stochastic process 
$\{X_s\}_{s\in\cS}$ is Gaussian w.r.t both $\PP_{\lam_1}$
and
$\PP_{\lam_2}$, 
due to Proposition
\ref{prop:16Oct19a}.
Then, using 
(\ref{eq:28Aug18b})
and
(\ref{eq:22Apr19g}), 
the mean value functions are given by
$$\wt{m}(s)=\EE_{\PP_{\lam_1}}\big[ X_s\big]=0
=\EE_{\PP_{\lam_2}}\big[
X_s\big]=m(s),\,
\forall s\in\cS$$
while from
(\ref{eq:23Oct19f}) 
the covariance functions are given by
\begin{multline*}
\wt{\Gamma}(s_1,s_2)=\Gamma_{\lam_1}(s_1,s_2)
=4^{-1}\lam_1^2
\big(\cN(s_1+s_2)-\cN(s_1-s_2)\big),
\\\Gamma(s_1,s_2)=
\Gamma_{\lam_2}(s_1,s_2)=
4^{-1}\lam_2^2
\big(\cN(s_1+s_2)-\cN(s_1-s_2)
\big),\end{multline*}
therefore 
$$\wt{\Gamma}(s_1,s_2)
=\lam_1^2\lam_2^{-2}
\Gamma(s_1,s_2).$$ 
Notice that the space
$\cH(\Gamma)$
consists of real functions on 
$\cS$,
which are continuous w.r.t the Fr\'echet topology
and $\cS$ is a separable metric space, thus we can use
\cite[Lemma 4.10]{H72} 
to justify the fact 
that
$\cH(\Gamma)$ is separable as well. Then, by letting 
$\{g_k\}_{k=1}^\infty$
be a complete orthonormal system in 
$\cH(\Gamma)$, 
one can write
$$\Gamma(s_1,s_2)
=\sum_{k=1}^\infty g_k(s_1)g_k(s_2)
\Longrightarrow 
\wt{\Gamma}(s_1,s_2)
=\sum_{k=1}^\infty 
\beta_kg_k(s_1)g_k(s_2)$$
where 
$\beta_k=
\lam_1^2
\lam_2^{-2}>0$ 
for all 
$k\ge1$. 
Finally, we apply Theorem 
\ref{thm:28Aug18c} 
to conclude that the measures 
$\PP_{\lam_1}$
and 
$\PP_{\lam_2}$ 
are equivalent if and only if 
$$\sum_{k=1}^\infty (1-\beta_k^2)=
\sum_{k=1}^\infty 
\big(1-\lam_1^4\lam_2^{-4}
\big)<\infty,$$
i.e., if and only if 
$\lam_1=\lam_2$.
Thus, whenever
$\lam_1\ne\lam_2$, the measures
$\PP_{\lam_1}$
and
$\PP_{\lam_2}$
are not equivalent, however due to  
\cite[Theorem 4.3]{J68}, 
it means they are mutually singular.
\end{proof} 
\begin{remark}
\label{rem:24Sep20b}
In the proof we use the same ideas as the first and second authors used in the proof of  
\cite[Corollary 7.5]{AJ14}.
\end{remark}
\begin{remark}
\label{rem:24Oct0a}
In
\cite{LPP}
the authors study a question that is quite different than our question, 
yet their condition for having a noise signal being detected or not is very similar to the condition we have; they use 
classical ideas
on convergence of measures and some 
analysis involved with the Radon--Nikodym
derivative, see for example
\cite[Lemma 3.3]{LPP}. 
For related results on Radon--Nikodym--Girsanov for Gaussian Hilbert spaces, see e.g., 
\cite{dVMU,QW}.
\end{remark}
Our present framework is motivated by applications. Indeed, a variety of families of mutually singular systems of path-space measures, and associated RKHSs
(see Theorems
\ref{thm:28Aug18d}
and 
\ref{thm:24Oct19a}) arise in diverse applications. The following includes a number of such distinct contexts: stochastic analysis, stochastic differential equations, analysis on fractals with scaling symmetry, and leaning theory models; see e.g., 
\cite{AJ12,AJ15}, 
\cite{J04}, 
and 
\cite{AMP05,AMP10}. 

To finish this section we present another interesting one-parameter family of CND functions on 
$\cS$ which reproduces a one-parameter family
of mutually singular measures on 
$\cS^\prime$, using the same machinery as in Theorem
\ref{thm:28Aug18d}.
\begin{example}
\label{ex:12Dec20a}
Fix a tempered measure $\mu\in\cM$ 
such that
$\cS\subset L_2(\mu)
\subset\cS^\prime$.
For every 
$u\in[0,1]$, we adapt the notations from Examples 
\ref{ex:13Dec20a} and 
\ref{ex:2Oct20a}, to 
define
$\cN_{\mu,u}:\cS\rightarrow\RR$ by
\begin{align}
\cN_{\mu,u}(s)=u\cN_{\mu}(s)+(1-u)\cN_{\wh{\mu}}(s)
=u\|s\|^2_{\mu}+(1-u)\|\wh{s}\|^2_{\mu},
\end{align}
that is a convex combination of
$\cN_{\mu,1}=\cN_{\mu}$
and
$\cN_{\mu,0}=\cN_{\wh{\mu}}$. 
It is easily seen that $\cN_{\mu,u}$ 
is a CND function which is continuous w.r.t the Fr\'echet topology
and satisfies condition
(\ref{eq:16Oct19b}), so there exists a probability measures 
$\PP_{\mu}^{(u)}$
on 
$\cS^\prime$ such that 
$\EE_{\PP_{\mu}^{(u)}}\big[\exp\{iX_s\}\big]=
\exp\{-\cN_{\mu,u}(s)/2\}$ for every 
$s\in\cS$.
With respect to the measure 
$\PP_{\mu}^{(u)}$,
the stochastic process
$\{X_s\}_{s\in\cS}$
is Gaussian (due to Proposition
\ref{prop:16Oct19a}) with the covariance function being equal to 
\begin{multline*}
\Gamma_{\mu,u}(s_1,s_2)=\EE_{\PP_{\mu}^{(u)}}[X_{s_1}X_{s_2}]=4^{-1}(\cN_{\mu,u}(s_1+s_2)-\cN_{\mu,u}(s_1-s_2))
\\=u\int_{\RR}s_1(t)s_2(t)d\mu(t)+(1-u)\int_{\RR}\wh{s_1}(t)\ol{\wh{s_2}(t)}
d\mu(t),
\end{multline*}
as it follows from formula
(\ref{eq:23Oct19f}) 
with $\lam=1$.
Let $\{h_n\}_{n\ge1}$ be an orthonormal basis of 
$L_2(\mu)$, therefore it is readily checked that 
\begin{align}
\label{eq:15Dec20a}
\Gamma_{\mu,u}(s_1,s_2)=
u\sum_{n=1}^\infty g_n(s_1)g_n(s_2)+(1-u)\sum_{n=1}^\infty
\wh{g_n}(s_1)\ol{\wh{g_n}(s_2)},
\end{align}
where $g_n(s)=\int_{\RR}s(t)h_n(t)d\mu(t)$ and $\wh{g_n}(s)=\int_{\RR}\wh{s}(t)h_n(t)d\mu(t)$ for every $s\in\cS$ and $n\ge 1$.
Under the assumption (which fails for 
$\mu$
being the Lebesgue measure on 
$\RR$) that  the system  
$\{u^{1/2}g_n\}_{n\ge1}\cup\{(1-u)^{1/2}\wh{g_n}\}_{n\ge1}$
is a Parseval frame (not necessarily orthogonal)
in 
$\cH(\Gamma_{\mu,u})$, 
we get that
\begin{multline*}
\Gamma_{\mu,v}(s_1,s_2)
\iffalse
=v\sum_{n=1}^\infty g_n(s_1)g_n(s_2)+(1-v)\sum_{n=1}^\infty\wh{g_n}(s_1)
\ol{\wh{g_n}(s_2)}
\fi
=vu^{-1}\sum_{n=1}^\infty 
u^{1/2}g_n(s_1)u^{1/2}g_n(s_2)+(1-v)(1-u)^{-1}\sum_{n=1}^\infty(1-u)^{1/2}\wh{g_n}(s_1)(1-u)^{1/2}\ol{\wh{g_n}(s_2)}
\end{multline*}
for every $v\in[0,1]$
and
$u\in(0,1)$ such that 
$v\ne u$, 
hence Theorem
\ref{thm:28Aug18c}
yields that the measures 
$\PP_{\mu}^{(u)}$
and
$\PP_{\mu}^{(v)}$ 
are mutually singular, as 
$$\sum_{n=1}^\infty 
\big(1-u^2v^{-2}\big)
+\sum_{n=1}^\infty 
\big(1-(1-v)^2(1-u)^{-2}\big)=\infty.$$
In conclusion, we obtain the one-parameter family 
$\{\PP_{\mu}^{(u)}\}_{0\le u\le 1}$ of mutually singular measures on 
$\cS^\prime$.
\end{example}
\begin{remark}
We used Theorem 
\ref{thm:28Aug18c} (that is Theorem 4.4 in
\cite{J68}) under the assumption that our system is only a Parseval frame system and not necessarily orthogonal complete system, as stated in the theorem itself.
Note that the time of Jorsboe's paper predates much later systematic studies of frame systems, i.e., the study of varieties of non-orthogonal 
expansions, which play an important rule in signal processing, spectral theory and wavelets theory;
see 
\cite{J06,HJW19,HJW19b}.
\end{remark}
\begin{remark}
As 
$\exp\{-\cN_{\mu,u}(s)/2\}=\exp\{-u\cN_{\mu}(s)/2\}\exp\{-(1-u)
\wh{\cN}_{\mu}(s)/2\}$
and in view of
(\ref{eq:14Oct19h}),
one can think of the space
$L_2(\cS^\prime,\PP_{\mu}^{(u)})$
as the tensor product 
$$L_s(\cS^\prime,\PP_{\mu}^{(u)})\cong L_2(\cS^\prime,\PP_{\mu,\sqrt{u}})
\otimes
L_2(\cS^\prime,\PP_{\wh{\mu},\sqrt{1-u}}),$$
where $\PP_{\mu,\sqrt{u}}$ and
$\PP_{\wh{\mu},\sqrt{1-u}}$
are the probability measures on 
$\cS^\prime$ obtained from 
(\ref{eq:14Oct19h}) by considering the initial measures $\PP_{\mu}$
and
$\PP_{\wh{\mu}}$
(instead of $\PP$), with the scalars 
$\lam=\sqrt{u}$
and
$\lam=\sqrt{1-u}$, respectively.  
\end{remark}
\iffalse
\begin{remark}
As the condition for the equivalence of the measures is expressed via...
$(\beta_k)$.. 
\end{remark}
\fi
\section{Isometric Isomorphisms and Intertwining Operators}
\label{sec:IsomIsom}
\subsection{The symmetric Fock space}
\label{subsec:fock}
In view of 
(\ref{eq:5Dec20a})
and
(\ref{eq:22Apr19f}), 
the mapping
$\psi$
on $\cH_0$ that is defined by
$\psi(\varphi_s)=X_s$,
satisfies
$$\|\psi(\varphi_s)\|^2_{L_2
(\cS^\prime,\PP_{\lam})}=\lam^2\cN(s)
=\lam^2
\|\varphi_s\|_{\cH}^2=
\|\varphi_s\|^2_{\cH_{\lam}}$$
for every 
$\lam\in\RR_+$
and hence --- as $\cH_0$ 
is a pre Hilbert space w.r.t the inner product given in
(\ref{eq:4Dec20c}), which is completed to the Hilbert space $\cH$ --- can be (uniquely) 
extended to an isometry between the Hilbert spaces 
$$\psi_{\lam}:\cH_{\lam}
\rightarrow L_2(\cS^\prime,
\PP_{\lam})$$
such that 
$\|\psi_{\lam}(h)\|_{L_2
(\cS^\prime,\PP_{\lam})}
=\|h\|_{\cH_{\lam}}$ 
for every 
$h\in \cH_{\lam}$.
%%%
We proceed by showing that the Hilbert space 
$\cH_{\lam}$ has an important 
role in our analysis of the 
$\lam-$white noise space, that is that 
$\Gamma_{sym}(\cH_{\lam})$ is isometrically
isomorphic to $L_2(\cS^\prime,\PP_{\lam})$ by an explicit transformation
given below. 
To do that we have to use the following lemma, which generalizes
a result from Hida's book \cite{H80} 
and will be used again later
in the
proof of Theorem
\ref{thm:24Oct19a}. 
\begin{lemma}
\label{lem:3May20a}
Let 
$\cN:\cS\rightarrow\RR$
be a CND function, that is
continuous w.r.t the Fr\'echet topology
and satisfies condition
(\ref{eq:16Oct19b}).
Then the set
$span\big(\exp\{X_s\}:s\in\cS\big)$
is dense in
$L_2(\cS^\prime,\PP_{\lam})$, 
for every $\lam\in\RR_+$.
\end{lemma}
A closer analysis on $L_2(\cS^\prime,\PP_{\lam})$, which includes the proof
of the lemma,
requires the Hermite polynomials; see 
(\ref{eq:3Dec20a})
for their precise definition and 
\cite{AJL11,AJ12,AJ15,AJL17}
for more details. 
\begin{theorem}
\label{thm:22Apr19a}
For every 
$\lam\in\RR_+$, 
the spaces
$L_2(\cS^\prime,\PP_{\lam})$ 
and 
$\Gamma_{sym}(\cH_{\lam})$ 
are isometrically isomorphic.
Moreover, an explicit 
isomorphism 
\begin{align}
\label{eq:8May20c}
W_{\lam}:\Gamma_{sym}
(\cH_{\lam})\rightarrow 
L_2(\cS^\prime,\PP_{\lam})
\end{align}
is presented in the proof in 
(\ref{eq:trans1}).
\end{theorem}
Before the proof, we recall some facts on the symmetric Fock space (cf. Section
\ref{sec:background}); 
for more supplementary facts see
\cite[Chapter 3-4]{Janson}. 
It is well known that the symmetric
Fock space
$\Gamma_{sym}(\cH_{\lam})$ 
is generated by the set
$span\big(\varepsilon(h):
h\in\cH\big)$, where 
$$\varepsilon(h):=
\sum_{n=0}^\infty
\frac{h^{\otimes n}}{\sqrt{n!}}.$$
Moreover, $\exp\{\langle h_1,
h_2
\rangle_{\cH_\lam}\}=\langle 
\varepsilon(h_1),
\varepsilon(h_2)
\rangle_{\Gamma_{sym}
(\cH_\lam)}$ 
for every 
$h_1,h_2\in\cH$
and in particular 
\begin{align}
\label{eq:14Oct19a1}
\| \varepsilon(h)\|_{
\Gamma_{sym}
(\cH_\lam)}^2
=\exp\{\|h\|^2_{\cH_\lam}\},
\quad\forall h\in\cH_\lam.
\end{align}
\begin{proof} 
As  
$\cH_0=span\big(\varphi_s:
s\in\cS\big)$ 
is dense in ($\cH$ and hence in)
$\cH_{\lam}$, 
we know that
$span\big(\varepsilon(\varphi_s):s\in\cS\big)$
is dense in 
$span\big(\varepsilon(h):h\in\cH_{\lam}\big)$
and hence
the symmetric Fock space
$\Gamma_{sym}(\cH_{\lam})$ is generated by 
$span\big(\varepsilon(\varphi_s):s\in\cS\big)$.
Therefore in order to define a mapping on 
$\Gamma_{sym}(\cH_{\lam})$, it is enough to define the mapping on these generators from
$\{\varepsilon(\varphi_s): 
s\in\cS\}$.
For any 
$s\in\cS$, 
define
\begin{align}
\label{eq:trans1}
W_{\lam}(\varepsilon(\varphi_s))
=\sum_{n=0}^\infty
(n!)^{-1}2^{-n/2}
X_{s}^n
=\exp\big\{X_s/\sqrt{2}\big\}.
\end{align}
Thus,
we have
\begin{multline*}
\big\|W_{\lam}(\varepsilon(\varphi_s))\big\|_
{L_2(\cS^\prime,\PP_{\lam})}^2
=\big\langle\exp\{X_s/\sqrt{2}\},
\exp\{X_s/\sqrt{2}\}
\big\rangle_{L_2(\cS^\prime,\PP_{\lam})}
=\int_{\cS^\prime}
\exp\{\sqrt{2}X_s(\omega)\}
d\PP_{\lam}(\omega)
\\=\sum_{n=0}^\infty
(n!)^{-1}2^{n/2}
\int_{\cS^\prime}X_s(\omega)^nd\PP_{\lam}(\omega)
=\sum_{m=0}^\infty
((2m)!)^{-1}2^m
\int_{\cS^\prime}X_s(\omega)^{2m}
d\PP_{\lam}(\omega)
\\+\sum_{m=0}^\infty
((2m+1)!)^{-1}2^{m+1/2}
\int_{\cS^\prime}
X_s(\omega)^{2m+1}d\PP_\lam(\omega),
\end{multline*}
and by applying equations 
(\ref{eq:28Aug18b})
and
(\ref{eq:22Apr19g}), we get 
\begin{align*}
\big\|W_{\lam}(\varepsilon(\varphi_s))\big\|_
{L_2(\cS^\prime,\PP_{\lam})}^2
=\sum_{m=0}^\infty 
((2m)!)^{-1}2^m(2m-1)!!
\|\varphi_s\|^{2m}_{\cH_\lam}
=\sum_{m=0}^\infty
(m!)^{-1}\|
\varphi_s\|_{\cH_\lam}^{2m}
\\=\exp\{\|\varphi_s\|_{\cH_\lam}^2\}
=\| \varepsilon(\varphi_s)\|^2_{\Gamma_{sym}(\cH_\lam)},
\end{align*}
when the last equality is due to 
(\ref{eq:14Oct19a1}).
Therefore for every $s\in\cS$, we have 
$$W_{\lam}(\varepsilon(\varphi_s))\in 
L_2(\cS^\prime,\PP_{\lam})
\text{ and } 
\|W_{\lam}(\varepsilon(\varphi_s))\|_
{L_2(\cS^\prime,\PP_{\lam})}
=
\|\varepsilon(\varphi_s)\|_{\Gamma_{sym}
(\cH_{\lam})}.$$
Next, as the set
$span\big(\varepsilon(\varphi_s)
:s\in\cS\big)$ generates 
$\Gamma_{sym}(\cH_{\lam})$, the mapping 
$W_{\lam}$ 
can be extended (uniquely) to a mapping
$W_{\lam}:\Gamma_{sym}(\cH_{\lam})
\rightarrow L_2(\cS^\prime,\PP_{\lam})$.
Finally, the extended mapping is onto 
$L_2(\cS^\prime,\PP_{\lam})$, 
as we know from Lemma 
\ref{lem:3May20a} that the set 
$span\big(\exp\{X_{s/\sqrt{2}}\}:
s\in \cS\big)=span
\big( \exp\{ X_s\}:
s\in\cS\big)$
generates the space $L_2(\cS^\prime,\PP_{\lam})$.
\end{proof} 
\begin{remark}
Another way to obtain an isometric isomorphism between 
$\Gamma_{sym}(\cH_{\lam})$
and
$L_2(\cS^\prime,\PP_{\lam})$,
as in 
\cite{JT16},
is by considering the mapping
$$\wt{W}_{\lam}(\varepsilon(\varphi_s)):
=\exp\big\{ X_s-
\|\varphi_s\|^2_{\cH_\lam}/2\big\},$$
while it is easily seen that 
$\big\|\wt{W}_{\lam}(\varepsilon(\varphi_s))
\big\|_{L_2(\cS^\prime,\PP_{\lam})}=\|\varepsilon(\varphi_s)\|_{\Gamma_{sym}(\cH_{\lam})}$
for every
$s\in\cS.$
\end{remark}
Not only that we have an isometric isomorphism
$W_{\lam}$ 
between $\Gamma_{sym}(\cH_{\lam})$
and
$L_2(\cS^\prime,\PP_{\lam})$,
for each 
$\lam\in\RR_+$, 
but we can also learn that some operators on 
one side become interesting operators on the other side.
The annihilation and creation operators on the symmetric Fock space (of 
$\cH_{\lam}$) 
are well studied and their relations and connections to physics, as well
as their geometric description, are well known; see for example
\cite[Chapter 6.3]{GlJa} 
and 
\cite{F89}.
On the other hand we can understand what they become after we use the
intertwining mapping 
$W_{\lam}$,
and then get a corresponding system of operators on 
$L_2(\cS^\prime,\PP_{\lam})$. 

The following proposition admits a nice relation between the stochastic properties in $L_2(\cS^\prime,\PP_{\lam})$ and the operator theoretic properties in
$\Gamma_{sym}(\cH_{\lam})$.
\begin{proposition}
The operator 
$W_{\lam}$
intertwines the pair of operators
$(\fa_{\lam},\fa_{\lam}^*)$
on $\Gamma_{sym}(\cH_{\lam})$ 
and the pair of operators 
$(\fd_{\lam},\fm_{\lam})$ 
on 
$L_2(\cS^\prime,\PP_{\lam})$, 
i.e.,
\begin{align}
\label{eq:24May20a}
\mathfrak{d}_{\lam}(h)=W_{\lam}\fa_{\lam}(h)W_{\lam}^*
\text{ and }
\fm_{\lam}(h)=W_{\lam}\fa_{\lam}(h)^*
W_{\lam}^*,\,\forall h\in\cH_{\lam}.
\end{align}
\end{proposition}
Here
$\{\fa_{\lam}(h),\fa_{\lam}^*
(h)\}_{h\in\cH_{\lam}}$
is the Fock representation of annihilation and creations operators on 
$\Gamma_{sym}(\cH_{\lam})$, 
while $\fm_{\lam}(h)$ 
stands for the multiplication operator by
$X_h$ and
$\mathfrak{d}_{\lam}(h)=\fm_{\lam}^*(h)$ 
stands for the abstract infinite dimensional
Malliavin derivative operator, 
both on the space 
$L_2(\cS^\prime,\PP_{\lam})$.
For a proof see 
\cite[Theorem 3.30 \& Section 4]{JT16}.
\begin{remark}
\label{rem:x}
The intertwining operator
$W_{\lam}$ 
is important,
because it makes the connection between two 
problems we study in this paper. 
The first problem is the mutual singularity of the Gaussian measures
$\{\PP_{\lam}\}_{\lam\in\RR_+}$ 
(cf. Subsection \ref{subsec:Jor}), 
while the second is  on representations of the CCR 
algebra being disjoint (cf. Subsection
\ref{subsec:CCR}). 
The link between these two problems is the intertwining operator 
$W_{\lam}$ 
(cf. the proof of Corollary
\ref{cor:23Sep20a}).
\end{remark}
\begin{remark}
\label{rem:3Jan21a}
There is a functor for symmetric Fock spaces, in the sense that
every operator on a one-particle Hilbert space, 
can be lifted to the symmetric Fock space of that Hilbert space.
Moreover, if the operator (on the one-particle) is contractive, 
then the corresponding lifted operator, called the second quantization
is a bounded operator on the symmetric Fock space.

This is applicable for us only when 
$\lam\le1$, when we let 
$S:\cH_{\lam}\rightarrow\cH_{\lam}$
be the operator defined as multiplication by $\lam$
--- which is now a contraction
--- 
thus the second quantization operator 
$\Gamma(S):\Gamma_{sym}(\cH_{sym})
\rightarrow\Gamma_{sym}(\cH_{sym})$, 
given by 
$\Gamma(S)(\varepsilon(h)):=\varepsilon(S(h))$,
is a bounded operator.
\end{remark}
\subsection{Generalized infinite Fourier transform}
\label{subsec:fourier}
In \cite[Chapter 4.3]{H80} 
the author presents an infinite dimensional 
generalized Fourier transform on the space 
$L_2(\cS^\prime,\wt{\PP})$,
where
$\wt{\PP}$ 
is the measure obtained from the Bochner--Minlos theorem
which corresponds in our settings 
to the case we choose the CND function to be
$\cN(s)=\|s\|^2_{L_2}.$
In this subsection we will follow the ideas from 
\cite{H80} 
and generalize some of the results 
to our case,
in which
$\cN:\cS\rightarrow\RR$
is a real CND function that 
is continuous w.r.t the Fr\'echet 
topology and satisfies condition
(\ref{eq:16Oct19b});
by doing so, we cover a much bigger
family of Gaussian processes. 

Another difference between our approach and the one in
\cite{H80} is that we define our generalized Fourier  transform
in a more direct way (cf. equation 
(\ref{eq:19Oct19d})), that is connected to the way we achieved  the Gaussian
process from an application of a Gelfand triple and the Bochner--Minlos theorem.
\\\\
{\em \uline{A reproducing kernel Hilbert space}.}
For every
$\lam\in\RR_+$,
recall that 
$$\cQ_{\lam}(s)=
\exp\big\{-\lam^2\cN(s)/2\big\}$$
is a positive definite function on 
$\cS$, which corresponds to the positive definite kernel
$Q_{\lam}(s_1,s_2)=\cQ_{\lam}(s_1-s_2)$
on
$\cS\times\cS$. 
There exists a (unique) RKHS, 
denoted 
$\cH(Q_{\lam}),$
which consists of real functions on 
$\cS$, where
$Q_{\lam}(\cdot,\cdot)$
is its RK. 
For every
$s\in\cS$, define the function
$Q_{\lam,s}\in
\cH(Q_{\lam})$
by
$$Q_{\lam,s}(\wt{s}):=
Q_{\lam}(s,\wt{s})=\cQ_{\lam}(s-\wt{s})
=\exp\{-\lam^2\cN(s-\wt{s})/2 \},
\forall \wt{s}\in\cS$$
and notice that 
$\cH(Q_{\lam})$ is 
generated by the subspace
$$\Big\{\sum_{j=1}^n c_j 
Q_{\lam,s_j}:
n\in\NN,\,
c_1,\ldots,c_n\in\RR,\,
s_1,\ldots,s_n\in\cS\Big\}.$$ 
Then for every 
$f\in \cH(Q_{\lam})$, we know that 
$\langle f, Q_{\lam,s}\rangle_{\cH(Q_{\lam})}=f(s);$
in particular
$$\langle Q_{\lam,s},Q_{\lam,\wt{s}}
\rangle_{\cH(Q_{\lam})}=Q_{\lam,s}(\wt{s})=Q_{\lam,\wt{s}}(s)=\cQ_{\lam}(s-\wt{s}).$$
{\em \uline{A generalized Fourier transform}.}
For every 
$F\in L_2(\cS^\prime,\PP_{\lam})$, 
we define a real  function
$\cT_{\lam}(F)$ on 
$\cS$, 
by the following rule
\begin{align}
\label{eq:19Oct19d}
(\cT_\lam F)(s):=\EE_{\PP_{\lam}}\big[
F\exp\{iX_s\}\big]=\int_{\cS^\prime}F(\omega)
\exp\{iX_s(\omega)\}d\PP_{\lam}(\omega)=\langle F, \exp\{i X_s\}
\rangle_{L_2(\cS^\prime,\PP_{\lam})}
\end{align} 
The transform 
$\cT_{\lam}$
is our analog of 
the classical Fourier transform, in infinite dimensions.
\begin{lemma}
\label{lem:23Sep20a}
For every 
$F\in
L_2(\cS^\prime,\PP_{\lam})$, 
we have
$\cT_{\lam}(F)\in
\cH(Q_{\lam})$.
\end{lemma}
To prove the lemma we use the following general fact 
(see
\cite{Ar50}): 
If 
$H$
is a RKHS of functions on a set 
$X$, 
with the RK
$K(\cdot,\cdot)$,
then $f:X\rightarrow\RR$ 
belongs to 
$H$ if and only if
there exists 
$C_f>0$ 
such that
for every
$n\in\NN,\,
c_1,\ldots,c_n\in\RR$, 
and 
$x_1,\ldots,x_n\in X$ 
we have
$$\Big|
\sum_{j=1}^n c_j f(x_j)
\Big|^2\le C_f\Big\|
\sum_{j=1}^n c_jK_{x_j}\Big\|_{H}^2 $$ 
where
$K_x\in H$ is defined by 
$K_x(\cdot)
=K(x,\cdot)$.
\begin{proof}
Let $F\in L_2(\cS^\prime,\PP_{\lam})$, then for every
$n\in\NN,\,
c_1,\ldots,c_n\in\RR$, and
$s_1,\ldots,s_n\in\cS$, we have
\begin{multline*}
\Big|
\sum_{j=1}^n
c_j\big(\cT_{\lam}(F)\big)(s_j)
\Big|^2=\Big|
\sum_{j=1}^n
c_j\EE_{\PP_{\lam}}
\big[ F\exp\{
iX_{s_j}\}\big]
\Big|^2
=\Big|
\EE_{\PP_{\lam}}\big[
F\sum_{j=1}^n
c_j
\exp
\{iX_{s_j}\}
\big]\Big|^2\\=\Big| 
\Big
\langle F,
\sum_{j=1}^n
c_j\exp\{iX_{s_j}\}
\Big\rangle_{L_2(\cS^\prime,\PP_{\lam})}
\Big|^2
\le
C_F\Big\|
\sum_{j=1}^n
c_j\exp\{iX_{s_j}\}
\Big\|^2_{L_2(\cS^\prime,\PP_{\lam})}
\end{multline*}by the Cauchy-Schwarz inequality,
where
$C_F=\|F\|^2_{L_2(\cS^\prime,
\PP_{\lam})}<\infty$. So,
\begin{multline*}
\Big|
\sum_{j=1}^n
c_j\big(\cT_{\lam}(F)\big)(s_j)
\Big|^2
\le
C_F\sum_{j,k=1}^n
c_jc_k
\EE_{\PP_{\lam}}
\big[\exp\{iX_{s_j-s_k}
\}\big]
=C_F\sum_{j,k=1}^n
c_jc_k
\cQ_{\lam}(s_j-s_k)
\\=C_F
\sum_{j,k=1}^n
c_jc_k
\langle Q_{\lam,s_j},
Q_{\lam,s_k}\rangle_{\cH(Q_{\lam})}
=C_F
\Big\| \sum_{j=1}^n
c_jQ_{\lam,s_j}\Big\|^2_{\cH(Q_{\lam})},
\end{multline*}
as needed, concluding 
$\cT_{\lam}(F)\in
\cH(Q_{\lam})$.
\end{proof}
In view of Lemma
\ref{lem:23Sep20a}, 
the image of 
$\cT_{\lam}$ is contained in $\cH(Q_{\lam})$. 
We next establish several further
properties of the operator
$\cT_{\lam}: L_2(\cS^\prime,
\PP_{\lam})\rightarrow
\cH(Q_{\lam})$
and its adjoint
$\cT_{\lam}^*:\cH(Q_{\lam})
\rightarrow L_2(\cS^\prime,\PP_{\lam})$,
which lead to the conclusion that $\cT_{\lam}$ is an isometric isomorphism.
\begin{theorem}
\label{thm:24Oct19a}
Fix $\lam\in\RR_+$.
\begin{itemize}
\item[(1)] 
For every $s\in\cS$, we have 
$\cT_{\lam}^*(Q_{\lam,s})=\exp\{-iX_s\}$
and
$\cT_{\lam}(\exp\{-iX_s\})=Q_{\lam,s}.$ 
\item[(2)]
$\cT_{\lam}^*$ 
is an isometry, i.e.,
$\|\cT_{\lam}^*(\psi)\|_{L_2(\cS^\prime,
\PP_{\lam})}=\|\psi\|_{\cH(Q_{\lam})}$
for every
$\psi\in \cH(Q_{\lam})$.
\item[(3)]
$\cT_{\lam}$ 
is onto $\cH(Q_{\lam})$, with 
$\ker(\cT_{\lam})=\{0\}.$ 
\end{itemize}
As both $\cT$ 
and
$\cT^*$
are isomorphic, we say that $\cT$
defines an isometric isomorphism from 
$L_2(\cS^\prime,\PP_{\lam})$
onto 
$\cH(Q_{\lam})$.
\end{theorem}
\begin{proof}
1. Fix
$s\in\cS$.
For every
$F\in L_2(\cS^\prime,\PP_{\lam})$, we have
\begin{align*}
\langle 
\cT_{\lam}(F),Q_{\lam,s}\rangle_{\cH(Q_{\lam})}
=\big(\cT_{\lam}(F)\big)(s)
=\EE_{\PP_{\lam}}
\big[ F\exp\{
iX_s\}\big]
=
\langle F,\exp\{
-iX_s\}\rangle_{L_2(\cS^\prime,\PP_{\lam})}
\end{align*}
which means that 
$\cT_{\lam}^*(Q_{\lam,s})
=\exp\{-iX_s\}$. 
Also, by the definition of 
$\cT_{\lam}$, 
for every 
$\wt{s}\in\cS$ we have
$$\big(\cT_{\lam}(\exp\{-iX_s\})\big)(\wt{s})
=\EE_{\PP_{\lam}}\big[ \exp\{iX_{\wt{s}-s}\}\big]
=\cQ_{\lam}(\wt{s}-s)=Q_{\lam,s}(\wt{s}),$$
i.e., $\cT_{\lam}(\exp\{-iX_s\})=Q_{\lam,s}$.
\\\\2.
For every 
$n\in\NN,\,
c_1,\ldots,c_n\in\RR$ and 
$s_1,\ldots,s_n\in \cS$, we have
\begin{multline*}
\Big\|\cT_{\lam}^*\Big(\sum_{j=1}^n
c_j Q_{\lam,s_j}\Big)\Big\|_{L_2(
\cS^\prime,\PP_{\lam})}^2
=\Big\|\sum_{j=1}^n
c_j\exp\{iX_{s_j}\}\Big\|_{L_2(\cS^\prime,
\PP_{\lam})}^2
\\=\sum_{j,k=1}^n
c_jc_k\cQ_{\lam}(s_j-s_k)=\Big\|\sum_{j=1}^n
c_jQ_{\lam,s_j}\Big\|^2_{\cH(Q_{\lam})},
\end{multline*}
however the subspace
$\Big\{\sum_{j=1}^n c_j 
Q_{\lam,s_j}:
n\in\NN,\,
c_1,\ldots,c_n\in\RR,\,
s_1,\ldots,s_n\in\cS\Big\}$
is dense in the Hilbert space 
$\cH(Q_{\lam})$ and hence we get that
$\|\cT_{\lam}^*(\psi)\|_{L_2(\cS^\prime,
\PP_{\lam})}=\|\psi\|_{\cH(Q_{\lam})}$
for every
$\psi\in\cH(Q_{\lam})$. That is of course equivalent to say that 
$\cT_{\lam}\cT_{\lam}^*=\mathcal{I}_{\cH(Q_{\lam})}$.
\\\\3.
Suppose that
$\psi\in\cH(Q_{\lam})$
is orthogonal to the image of 
$\cT_{\lam}$ on 
$L_2(\cS^\prime,\PP_{\lam})$. 
Then
$0=\big\langle 
\psi,
\cT_{\lam}(\exp\{-iX_s\})
\big\rangle_{\cH(Q_{\lam})}
=\langle \psi, Q_{\lam,s}
\rangle_{\cH(Q_{\lam})}=\psi(s)$
for every
$s\in\cS$,  
i.e., $\psi=0$, which means that  the image of 
$\cT_{\lam}$
on
$L_2(\cS^\prime,\PP_{\lam})$ 
is equal to 
$\cH(Q_{\lam})$.

Finally, let 
$F\in\ker(\cT_{\lam})$, 
then 
$\langle F, \exp\{i X_s\}
\rangle_{L_2(\cS^\prime,\PP_{\lam})}
=(\cT_{\lam}(F))(s)
=0$ 
for any
$s\in\cS$,
i.e., $F$
is orthogonal to the subspace
$\{\exp\{iX_s\}:s\in\cS\}$
of
$L_2(\cS^\prime,\PP_{\lam})$. 
However, it follows from Lemma
\ref{lem:3May20a}, 
that this subspace is dense in $L_2(\cS^\prime,\PP_{\lam})$,
thus 
$F=0$. 
\end{proof}
\begin{remark}
\label{rem:24Sep20c}
Some of the results in this subsection, 
such as the formula for the Fourier transform and its being isometrically 
isomorphism, can be obtained from 
\cite[Section 3.4]{AJL17} 
as well.
\end{remark}
\begin{remark}
This kind of infinite dimensional Fourier transform is used in applications to statistics as well, see e.g.
\cite{EPS09,SS13}
where the studied transform is the Esscher transform.
\end{remark}
We built two isometric isomorphisms 
$$W_{\lam}:\Gamma_{sym}(\cH_{\lam})\rightarrow
L_2(\cS^\prime,\PP_{\lam})
\text{ and }
\cT_{\lam}: L_2(\cS^\prime,\PP_{\lam})\rightarrow
\cH(Q_{\lam}),$$ 
both involve with the Gaussian measure $\PP_{\lam}$. 
However, it is only natural to take the composition of these two isometric
isomorphisms, that is
\begin{align}
\label{eq:23May20f}
R_{\lam}:\Gamma_{sym}(\cH_{\lam})
\rightarrow\cH(Q_{\lam}),\,
R_{\lam}:=\cT_{\lam}\circ W_{\lam}
\end{align}
which is an isometric isomorphism by itself. 
\begin{remark}
Both the symmetric Fock space 
$\Gamma_{sym}(\cH_{\lam})$
and the RKHS 
$\cH(Q_{\lam})$
are obtained directly from the CND function $\cN:\cS\rightarrow\RR$ that
we start from, independently of the measure $\PP_{\lam}$. 
Hence this mapping $R_{\lam}$ described below appears more naturally in the context of 
\cite{J68}.
\end{remark}
\begin{cor}
The spaces
$\Gamma_{sym}(\cH_{\lam})$ and 
$\cH(Q_{\lam})$ 
are isometrically isomorphic, via the 
explicit isometric 
isomorphism $R_{\lam}$ between the two, that 
is given in 
(\ref{eq:24Sep20e}). 
\end{cor}
In the proof we get an explicit formula for the mapping, while using the complexification of the Schwartz class $\cS$ --- 
which follows the same ideas as in
\cite[Chapter 6.2]{H80} and is possible due to 
(\ref{eq:22Jul20a}) ---
and in particular the relation $X_{s+is^\prime}:=X_s+iX_{s^\prime}$ 
for every 
$s,s^\prime\in\cS$.
\begin{proof}
Clearly, 
$R_{\lam}:=\cT_{\lam}\circ W_{\lam}$
is an isometric isomorphism, as both $W_{\lam}$
and
$\cT_{\lam}$
are. Moreover, we have the following explicit way to write its formula.
For every $s\in\cS$, the function 
$$R_{\lam}(\varepsilon(\varphi_s))
=\cT_{\lam}\big(W_{\lam}(\varepsilon(\varphi_s)\big)
\in
\cH(\cQ_{\lam})$$ 
satisfies 
\begin{align*}
\big(R_{\lam}(\varepsilon
(\varphi_s))\big)(s^\prime)
=\Big(\cT_{\lam}
\big(\exp\big\{
X_s/\sqrt{2}\big\}\big)
\Big)(s^\prime)
=\int_{\cS^\prime}
\exp\{
X_s/\sqrt{2}\}\exp\{
iX_{s^\prime}\}d\PP_{\lam}
\\=\int_{\cS^\prime}
\exp\{iX_{s^\prime-is/\sqrt{2}}\}
d\PP_{\lam}
=\cQ_{\lam}(s^\prime-is/\sqrt{2})
=Q_{\lam,is/\sqrt{2}}(s^\prime),
\end{align*} 
for every 
$s^\prime\in\cS$. 
Thus we have the explicit formula (on basis elements), which is
\begin{align}
\label{eq:24Sep20e}
R_{\lam}(\varepsilon(\varphi_s))=Q_{\lam,is/\sqrt{2}}.
\end{align} 
\end{proof}
\subsection{The canonical commutation relations}
\label{subsec:CCR}
Let 
$\cL$ be a Hilbert space.
The algebra
$CCR(\cL)$ 
is generated axiomatically by
a system
$\{\fa(\ell),\fa^*(\ell)\}_{\ell\in\cL}$, whereas $\fa(\ell),\fa^*(k)$ are
operators (on an Hilbert space),
subject to
\begin{align}
\label{eq:8May20a}
[\fa(\ell),\fa(k)]
=0,\quad
\forall \ell, k\in \cL
\end{align}
and
\begin{align}
\label{eq:8May20b}
[\fa(\ell),\fa^*(k)]=
\langle\ell,k\rangle_{\cL}.
\end{align}
A system
$\{\fa(\ell),\fa^*(\ell)\}_{\ell\in\cL}$ 
which satisfies 
(\ref{eq:8May20a})
and
(\ref{eq:8May20b}) 
is called a \textit{representations of the CCR algebra} $CCR(\cL)$. 
Some representations, such as 
the Fock representations, might be realized in  the symmetric Fock space
$\Gamma_{sym}(\cL)$; 
e.g., see Example
\ref{ex:23Sep20a}.
\begin{definition}
Let 
$\{\fa(\ell),\fa^*(\ell)\}_{\ell\in\cL}$
and
$\{\fb(\ell),\fb^*(\ell)\}_{\ell\in\cL}$
be two representations of 
$CCR(\cL)$, w.r.t the Hilbert spaces $\cL_1,\cL_2$ (i.e., $\fa(\ell):\cL_1\rightarrow\cL_1$
and 
$\fb(\ell):\cL_2\rightarrow\cL_2$
for every
$\ell\in\cL$). We say that the representations are 
\textbf{unitarily equivalent} if there exists a unitary operator 
$U:\cL_1\rightarrow\cL_2$ such that 
$U\fa(\ell)U^{-1}=\fb(\ell)$ 
for every 
$\ell\in\cL$.
\end{definition}
Segal's work
(\cite{S56a,S56b,S81}) 
was motivated by quantum mechanics, while
Bargmann motivation for the case of finite number of degrees of freedom was
in complex analysis (of multivariable complex functions). The result by Stone
and von Neumann (see 
\cite{St30,St32,vN31,vN32})treats this case of finite number of degrees of freedom,
that is the uniqueness --- up to unitarily equivalence --- of families of unitary
operators which are irreducible and satisfy the Weyl commutation relations.
\begin{theorem}[Stone non-Neumann]
\label{thm:StvN}
If $\cL$ is a finite dimensional Hilbert space, then any two representations
of the CCR algebra 
$CCR(\cL)$ 
are unitarily equivalent.
\end{theorem}
In his book
\cite{H80}, Hida gives a nice survey proof for this fact and an explantation
to what happens when going to infinite number of degrees of freedom. 
The question of how bad can it be when one goes to an infinite dimensional
Hilbert space arises. It turns out that in addition to the case of equivalent
representations, there is one (and only one) more option and that is that
the representations are disjoint. 
\begin{definition}
Two representations of the CCR algebra are 
said to be \textbf{disjoint}, 
if the only operator that intertwines one with the other is zero. 
In other words, two representations of 
$CCR(\cL)$, say
$\{\fa(\ell),\fa^*(\ell)\}_{\ell\in\cL}$
and
$\{\fb(\ell),\fb^*(\ell)\}_{\ell\in\cL}$
are disjoint if for every unitary operator 
$U:\cL_1\rightarrow\cL_2$ such that 
$U\fa(\ell)=\fb(\ell)U$
for every $\ell\in\cL$, we must have $U=0$.
\end{definition}
\begin{theorem}
\label{thm:2Oct20a}
If
$\cL$ is infinite dimensional, then any two irreducible representations of
the CCR algebra
$CCR(\cL)$ are either equivalent or disjoint.
\end{theorem}
\begin{remark}
The Stone von-Neumann theorem 
(cf. Theorem \ref{thm:StvN}) 
and Theorem 
\ref{thm:2Oct20a} 
illustrate one difference between the finite and infinite dimensional cases.
So is the result in
Theorem
\ref{thm:28Aug18d}, regarding the mutual singularity of the measures, which is true only in the infinite dimensional case.
Another difference is that the group of unitary operators 
on the Hilbert space 
$\cL$ induces a group of measure preserving transformations on $L_2$ space of a Gaussian, and that action is ergodic if and only if 
$\cL$
is of infinite dimension. This explained nicely in 
\cite{SS68,SS70} and goes back to ideas of Segal in the papers
\cite{S56a,S56b}.
\end{remark}
There are many ways to construct unitarily inequivalent representations of
the CCR (w.r.t an infinite dimensional Hilbert space). 
One example of an explicit way of finding infinitely many unitarily 
inequivalent representations is given in
\cite{S56a,S56b} 
and presented (shortly) here: 
\begin{example}
\label{ex:23Sep20a}
On the symmetric Fock space 
$\Gamma_{sym}(\cH)$
we define the operators
\begin{align*}
A_j:\Gamma_{sym}(\cH)
\rightarrow \Gamma_{sym}(\cH),\,
A_j(E_{\alpha})=
\sqrt{\alpha_j}E_{\alpha-1_j},\,
\forall \alpha\in\ell_0^{\NN}
\end{align*}
and
\begin{align*}
A_j^*:\Gamma_{sym}(\cH)
\rightarrow 
\Gamma_{sym}(\cH),\,
A_j^*(E_{\alpha})
=\sqrt{\alpha_j+1}E_{\alpha+1_j},\,
\forall \alpha\in\ell_0^{\NN}
\end{align*}
where the basis elements 
$\{E_{\alpha}\}_{\alpha\in\ell_0^{\NN}}$
are defined in
(\ref{eq:5Dec20c}).
\iffalse
One can easily check that 
$A_j^*$ 
has the form
$A_j^*\Big(\sum_{k=0}^\infty 
\nu^{(k)}\Big)
=\sum_{k=0}^\infty 
\sqrt{k+1}S(e_j\otimes \nu^{(k)}),\,
\text{where }\nu^{(k)}
\in\cH^{\otimes k},\,
\forall k\ge0.$
\fi
The family 
$\{A_j,A_j^*\}_{j=1}^{\infty}$
is a representation of the CCR algebra and
for any 
$j\ge1$ 
the operators 
$A_j$
and
$A_j^*$
are adjoints in 
$\Gamma_{sym}(\cH)$.
Using the family 
$\{A_j,A_j^*\}_{j=1}^\infty$,
we can build another family of 
operators depending on a scalar
$\lam\in\RR_+$, 
in the following way- let
$$A_{j,\lam}:\Gamma_{sym}
(\cH)\rightarrow\Gamma_{sym}(\cH),\,
A_{j,\lam}(E_{\alpha})
=\lam\sqrt{\alpha_j}
E_{\alpha-1_j},\,
\forall \alpha\in\ell_0^{\NN}$$
and
$$
A_{j,\lam}^*:\Gamma_{sym}(\cH)
\rightarrow\Gamma_{sym}(\cH),\,
A_{j,\lam}^*(E_{\alpha})
=\frac{\sqrt{\alpha_j+1}}
{\lam}E_{\alpha+1_j},\,
\forall \alpha\in\ell_0^{\NN},$$
which are actually given by 
$A_{j,\lam}=\lam A_j$ 
and 
$A_{j,\lam}^*=\lam^{-1} A_j^*$. 
The family 
$\{A_{j,\lam},A_{j,\lam}^*\}_{j=1}^\infty$
is also a representation of the CCR algebra and in Segal's papers it is shown that this representation is unitarily inequivalent to
the representation $\{A_j,A_j^*\}_{j=1}^\infty$, as long as $\lam\ne1$.
\end{example}
\iffalse
\begin{example}
Segal was stressing out the fact that if $\{P_j,Q_j\}$ is a representation
of the CCR, i.e., if
$[P_j,Q_k]=-i\delta_{j,k}I,
[P_j,P_k]=[Q_j,Q_k]=0$
and the operators are self adjoint acting on an Hilbert space
$\cH$, then 
$\{P_j^\lam,Q_j^{\lam}\}$ is a representation of the 
CCR as well, where
$\lam\in\RR_+,
P_j^\lam=\lam P_j$
and
$Q_j^{\lam}=\lam^{-1}Q_j$.\\
Given such
$\{P_j,Q_j\}$, for every pair 
$(\ul{a},\ul{b})\in\RR^N\times\RR^N$
corresponds an operator
$$W(\ul{a},\ul{b})=e^{i \ul{a}\cdot P}e^{i\ul{b}\cdot Q},$$
where 
$\ul{a}\cdot P=\sum_{k=1}^N 
a_kP_k$
and
$\ul{b}\cdot Q=
\sum_{k=1}^N
b_kQ_k$.
These operators satisfy the relations
$$W(h)W(h^\prime)=e^{i\langle h,h^\prime\rangle}W(h+h^\prime),$$ 
where
$h=(\ul{a},\ul{b})$
and
$h^\prime=(\ul{a}^\prime,\ul{b}^\prime)$
are in
$\RR^N\times \RR^N$.
%%
In
\cite{},
Segal proved that under the simple transformation of multiplication by 
$\lam\in\RR_+$, the obtained new representation of the CCR is unitarily inequivalent
to the original representation, i.e., the two families
$\{W(h)\}_{h\in\RR^N\times\RR^N}$
and
$\{W_\lam(h)\}_{h\in\RR^N\times\RR^N}$ are unitarily inequivalent, where
$W_\lam(h)$ corresponds to the family of operators
$\{P^\lam_j,Q^\lam_j\}$.
\end{example}
\fi
In many works the connections between the equivalence of the symmetric Fock
spaces and the equivalence of the prospective measures are stressed out.
We finish this paper by  presenting a family of unitarily
inequivalent representations of the CCR, 
which follows from the influence of the simple operation of multiplication by $\lam$, on our  white noise space, symmetric Fock space and the corresponding RKHS, which appear in earlier stages of the paper.
\begin{cor}
\label{cor:23Sep20a}
Let
$\pi=\{\fa(h),\fa^*(h)\}_{h\in\cH}$
be a Fock representation of the CCR over $\cH$.
For every $\lam\in\RR_+$ 
define the $\lam-$Fock representation 
\begin{align}
\label{eq:23May20a}
\pi_{\lam}:=\{\fa_{\lam}(h),\fa_{\lam}(h)^*
\}_{h\in\cH_\lam},
\text{ with }
\fa_{\lam}(h)=\lam\fa(h)
\text{ and }\fa_{\lam}^*(h)=\lam^{-1}\fa^*(h).
\end{align}
Then any two representations from the  family 
$\big\{
\pi_{\lam}:\lam\in\RR_+\big\}$
are disjoint, i.e., 
$\pi_{\lam}$
and
$\pi_{\lam^\prime}$
are disjoint (irreducible) representations for every
$\lam,\lam^\prime\in\RR_+$ with 
$\lam\ne\lam^\prime$.
\end{cor}
The idea of the proof is that 
there are two notions of equivalence; 
one is
equivalence of Gaussian measures and another 
is equivalence of irreducible
representations. 
There is a non trivial theorem saying that the two notions of equivalence
are equivalent. 
It is hinted (not explicitly) 
in several places in
\cite{GlJa,P74,N67,H80,JT16,Janson}, 
so we only sketch some of the ideas of 
the proof of Corollary 
\ref{cor:23Sep20a}.
\begin{proof}
Let
$\lam,\lam^\prime\in\RR_+$ 
be such that 
$\lam\ne\lam^\prime$. Suppose that the representations 
$\pi_{\lam}$
and
$\pi_{\lam^\prime}$
are unitarily equivalent, i.e., suppose there exists a unitary mapping
\begin{align}
\label{eq:23May20c}
U:\Gamma_{sym}(\cH_{\lam})
\rightarrow\Gamma_{sym}
(\cH_{\lam^\prime})
\end{align}
such that 
\begin{align}
\label{eq:23May20d}
U\fa_{\lam}(h)=
\fa_{\lam^\prime}(h)U
\text{ and }
U\fa_{\lam}^*(h)
=\fa_{\lam^\prime}^*(h)U,\,
\forall h\in\cH.
\end{align}
\uline{\textbf{Step 1: Obtain a family of representations (of operators on
$L_2(\cS^\prime)$):}}
Recall the mapping 
$$W_{\lam}:\Gamma_{sym}(\cH_{\lam})\rightarrow L_2(\cS^\prime,\PP_{\lam})$$
that is an isometric isomorphism (cf. Theorem
\ref{thm:22Apr19a}), with
\begin{align*}
W_{\lam}\fa_{\lam}(h)
=\fd_{\lam}(h)W_{\lam}
\text{ and }
W_{\lam}\fa_{\lam}^*(h)
=\fm_{\lam}(h)W_{\lam},
\end{align*}
where $\fd_{\lam}$
and
$\fm_{\lam}$
are given in
(\ref{eq:24May20a}).
Define the mapping 
\begin{align}
\label{eq:23May20b}
V=W_{\lam^\prime}UW_{\lam}^*
:L_s(\cS^\prime,
\PP_{\lam})\rightarrow 
L_2(\cS^\prime,\PP_{\lam^\prime}),
\end{align}
then $V$ 
is unitary and satisfies the intertwining relations
\begin{align}
\label{eq:24May20b}
V\fm_{\lam}(h)=\fm_{\lam^\prime}(h)V 
\text{ and }
V\fd_{\lam}(h)=\fd_{\lam^\prime}(h)V,
\,\forall h\in h\in\cH.
\end{align}
The justification for 
(\ref{eq:24May20b})
is because
\begin{multline*}
V\fm_{\lam}(h)
=\big(W_{\lam^\prime}U
W_{\lam}^*\big)
\fm_{\lam}(h)W_{\lam}
W_{\lam}^*
=W_{\lam^\prime}U\big(
W_{\lam}^*\fm_{\lam}(h)W_{\lam}\big)
W_{\lam}^*
=W_{\lam^\prime}
\big(U
\fa^*_{\lam}(h)
\big)W_{\lam}^*
\\=W_{\lam^\prime}
\big(\fa^*_{\lam^\prime}(h)U
\big)W_{\lam}^*
=W_{\lam^\prime}
\big(W_{\lam^\prime}^*
\fm_{\lam^\prime}(h)
W_{\lam^\prime}\big)
UW_{\lam}^*
=\fm_{\lam^\prime}(h)V,
\end{multline*}
and thus by taking the adjoints and then multiplying by $V$ on both sides
we get
\begin{align*}
\fm_{\lam}(h)^*V^*=V^*\fm_{\lam^\prime}^*(h)\Longrightarrow  
V\fm_{\lam}^*(h)=\fm_{\lam^\prime}^*(h)V
\Longrightarrow V\fd_{\lam}(h)=\fd_{\lam^\prime}(h)V.
\end{align*}
Thus instead of looking at the intertwining operator 
$U$ 
and the representations 
$\pi_{\lam}$
and
$\pi_{\lam^\prime}$,
we can consider
the intertwining operator 
$V$ 
and the representations 
$\Pi_{\lam}$ and
$\Pi_{\lam^\prime}$, which are given by
\begin{align}
\label{eq:23May20e}
\Pi_{\lam}=\big\{\fd_\lam(h),
\fm_{\lam}(h)
\big\}_{h\in\cH_{\lam}},
\text{ where }
\fd_{\lam}(h)=W_{\lam}
\fa_{\lam}(h)W_{\lam}^*
\text{ and }\fm_{\lam}(h)
=W_{\lam}\fa^*_{\lam}
(h)W_{\lam}^*
\end{align}
and
\begin{align}
\label{eq:23May20e}
\Pi_{\lam^\prime}=
\big\{\fd_{\lam^\prime}(h),
\fm_{\lam^\prime}(h)
\big\}_{h\in\cH_{\lam^\prime}},
\text{ where }
\fd_{\lam^\prime}(h)=W_{\lam^\prime}
\fa_{\lam^\prime}(h)W_{\lam^\prime}^*
\text{ and }\fm_{\lam^\prime}(h)
=W_{\lam^\prime}\fa^*_{\lam^\prime}
(h)W_{\lam^\prime}^*
\end{align}
notice that 
$\fd_{\lam}(h)$ and
$\fm_{\lam}(h)$ 
are operators on the space
$L_2(\cS^\prime,\PP_{\lam})$.
\\\\\uline{\textbf{Step 2:}}
We obtained that the multiplication operators 
$\{\fm_{\lam}(h)\}_{h\in\cH_{\lam}}$
on the space
$L_2(\cS^\prime,\PP_{\lam})$
are equivalent to the multiplication operators
$\{\fm_{\lam^\prime}(h)\}_{h\in\cH_{\lam^\prime}}$
on the space 
$L_2(\cS^\prime,\PP_{\lam^\prime})$,
in the sense of the existence of an intertwining operator between these families.
This implies that the measures $\PP_{\lam}$
and
$\PP_{\lam^\prime}$
must have a Radon--Nikodym derivative, see for example 
\cite[Chapter 6]{N67}.
However the two measures 
$\PP_{\lam}$
and
$\PP_{\lam^\prime}$ on 
$\cS^\prime$
are singular, as $\lam\ne\lam^\prime$ 
(cf. Theorem \ref{thm:28Aug18d}),
therefore we obtained a contradiction, as needed.
\\\\\uline{\textbf{Step 3:}}
From steps $1$ and $2$ it follows that  the representations 
$\pi_{\lam}$
and
$\pi_{\lam^\prime}$
are unitarily inequivalent whenever 
$\lam\ne\lam^\prime$. To conclude that the two representations are disjoint, we will show they are both irreducible and then use Theorem 
\ref{thm:2Oct20a}. 
The proof of the irreducibility of the representation
$\pi_{\lam}$ can be obtained by repeating steps $1$ and $2$,  for the case where $\lam=\lam^\prime$ and eventually use the fact that the dual pair combined system $\{\fm_{\lam}(h),\fd_{\lam}(h)\}_{h\in\cH_{\lam}}$
is irreducible and so any operator which commutes with the combined system of multiplication operators and their duals must be  a scalar operator 
(the system
$\{\fm_{\lam}(h)\}_{h\in\cH_{\lam}}$ 
is a maximal abelian algebra of multiplication operators, hence the only operator which commutes with those multiplication operators must be a multiplication operator by itself, but then as it also commutes with the system of derivatives
$\{\fd_{\lam}(h)\}_{h\in\cH_{\lam}}$, it must be a scalar operator).
 \iffalse
 it follows from ideas taken from
\cite{Hida}, 
\cite[Chapter 6]{Nelson},
\cite[Example..]{JoTi}, and
the GNS construction in 
\cite{Powers}, 
that the only intertwining operator between these two multiplication operators
is the zero operator, which is a contradiction to $V$ being a unitary operator.

It relies on the idea of identifying representations of the CCR algebras
by positive linear functionals, see
\cite{Powers} for example.
\fi
\end{proof}
\textbf{Acknowledgment.} 
Daniel Alpay thanks the Foster G. and Mary McGaw Professorship in
Mathematical Sciences, which supported this research. 
Palle Jorgensen and Motke Porat thank Chapman University for their hospitality for many fruitful visits over the last $3$ 
years. 
Motke Porat thanks the University of Iowa for a one week visit, which was very helpful for the creation of this paper. 
Lastly, the research of the third named author was partially supported by the US-Israel Binational Science Foundation (BSF)
Grant No. 2010432, 
Deutsche Forschungsgemeinschaft (DFG) 
Grant No. SCHW 1723/1-1, 
and Israel Science
Foundation (ISF) Grant No. 2123/17.

%-Bibliography

\end{document}